\documentclass[12pt]{article}%
\usepackage[affil-it]{authblk}
\usepackage{amssymb,amsthm}
 \usepackage{url}
\usepackage{mathtools}
\usepackage{mathabx}
\usepackage[multiple]{footmisc}
\usepackage{enumerate}
\usepackage[textwidth=2cm]{todonotes}
\usepackage{xcolor}
\usepackage{epic}
\usepackage{tikz}
\usepackage{setspace}
\usepackage{enumitem}
\usetikzlibrary{decorations.pathreplacing}
\usepackage{graphicx}%
\usepackage{relsize}
\usepackage{comment}

 \usepackage{xr-hyper}
 \usepackage{soul}
 \usepackage{calrsfs}
 
\usepackage[round]{natbib}

\makeatletter
\def\BState{\State\hskip-\ALG@thistlm}
\makeatother

\makeatletter
\newcommand{\mylabel}[2]{#2\def\@currentlabel{#2}\label{#1}}
\makeatother

\usepackage{amsmath}
\usepackage{amssymb,amsfonts}
\usepackage{amsthm}
\allowdisplaybreaks

 \DeclareMathAlphabet{\pazocal}{OMS}{zplm}{m}{n}

 \newcommand{\Fp}{\pazocal{F}} 
  \newcommand{\Eq}{\pazocal{E}} 
  \newcommand{\NE}{\pazocal{N}}

\definecolor{clemson-orange}{RGB}{234,106,32}
\definecolor{chicago-maroon}{RGB}{128,0,0}
\definecolor{northwestern-purple}{RGB}{82,0,99}
\definecolor{cornell-red}{RGB}{179,27,27}
\definecolor{sauder-green}{RGB}{171,180,0}
\definecolor{gray}{RGB}{192,192,192}
\definecolor{lawngreen}{RGB}{0,250,154}
 
\usepackage{fullpage}
\usepackage[colorlinks,citecolor=northwestern-purple,urlcolor=chicago-maroon,linkcolor=chicago-maroon,filecolor=chicago-maroon,backref=page]{hyperref}

\usepackage[nameinlink]{cleveref}
 \usepackage{graphicx}

\providecommand{\U}[1]{\protect\rule{.1in}{.1in}} \textheight 8.2in
\def\mG{\pazocal{G}}
\newtheorem{theorem}{Theorem}
\newtheorem{lemma}{Lemma}
\newtheorem{corollary}{Corollary}

\newtheorem{proposition}{Proposition}

\newtheorem{claim}{Claim}
\newtheorem{fact}{Fact}

\newtheorem{manualtheoreminner}{Theorem}

\newenvironment{manualtheorem}[1]{%
  \renewcommand{\themanualtheoreminner}{#1}%
  \manualtheoreminner
}{\endmanualtheoreminner}

\newtheorem{manualpropinner}{Proposition}

\newenvironment{manualprop}[1]{%
    \renewcommand{\themanualpropinner}{#1}%
  \manualpropinner
}{\endmanualpropinner}

\newtheorem{manualcorollaryinner}{Corollary}

\newenvironment{manualcorollary}[1]{%
    \renewcommand{\themanualcorollaryinner}{#1}%
  \manualcorollaryinner
}{\endmanualcorollaryinner}

\def\r{\succeq}

\theoremstyle{definition}

\newtheorem{example}{Example}
\newtheorem*{algorithm*}{Greedy Algorithm}

\definecolor{darkgreen}{rgb}{0.0, 0.5, 0.0}
\definecolor{brightpink}{rgb}{1.0, 0.0, 0.5}
\definecolor{brightgreen}{rgb}{0.4, 1.0, 0.0}
\definecolor{shockingpink}{rgb}{0.99, 0.06, 0.75}
\definecolor{persiangreen}{rgb}{0.0, 0.65, 0.58}

\def\bq{\begin{equation}}
\def\eq{\end{equation}}
\def\ba{\begin{eqnarray}}
\def\ea{\end{eqnarray}}
\def\bas{\begin{eqnarray*}}
\def\eas{\end{eqnarray*}}

\def\G{\Gamma}

\def\M {\mathcal M}

\def\x{\mathbf{x}}
\def\y{\mathbf{y}}

\def\M{\mathcal{M}}

\def\M{\mathcal{M}}
\def\U{\Upsilon}

\def\bh{\overline b}
\def\bl{\underline b}
\def\sh{\overline {\sigma}}
\def\sl{\underline {\sigma}}

\def\midcup{\mathsmaller{\bigcup}}

\makeatletter
\def\mathcenterto#1#2{\mathclap{\phantom{#1}\mathclap{#2}}\phantom{#1}}
\let\old@widetilde\widetilde
\def\widetildeto#1#2{\mathcenterto{#2}{\old@widetilde{\mathcenterto{#1}{#2\,}}}}
\let\old@widehat\widehat
\def\widehatto#1#2{\mathcenterto{#2}{\old@widehat{\mathcenterto{#1}{#2\,}}}}
\makeatother

\def\join{\curlyvee}
\def\meet{\curlywedge}

\newcommand{\Meet}{\mathop{\vcenter{\hbox{\scalebox{1.2}[1.7]{$\curlywedge$}}}}}
\renewcommand{\Join}{\mathop{\vcenter{\hbox{\scalebox{1.2}[1.7]{$\curlyvee$}}}}}

\definecolor{clemson-orange}{RGB}{234,106,32}
\definecolor{chicago-maroon}{RGB}{128,0,0}
\definecolor{northwestern-purple}{RGB}{82,0,99}
\definecolor{cornell-red}{RGB}{179,27,27}
\definecolor{sauder-green}{RGB}{171,180,0}
\definecolor{gray}{RGB}{192,192,192}
\definecolor{lawngreen}{RGB}{0,250,154}

\newcommand{\jw}[1]{\textcolor{northwestern-purple}{{[JW: #1]}}}
\newcommand{\fk}[1]{\textcolor{chicago-maroon}{{[Fuhito: #1]}}}

\allowdisplaybreaks

\usepackage{geometry}
\begin{document}

\title{\bf Monotone Comparative Statics without Lattices}

\author{\textsc{Yeon-Koo Che},\thinspace\thinspace\ \textsc{Jinwoo Kim},\thinspace\thinspace\ \textsc{Fuhito Kojima}\thanks{Che: Department of Economics, Columbia University (email: yeonkooche@gmail.com); Kim: Department of Economics, HKUST (email: jikim72@gmail.com); Kojima: Department of Economics and Market Design Center, the University of Tokyo (email: fuhitokojima1979@gmail.com). The authors are grateful to the editor and the referees for their thoughtful comments and constructive suggestions, which have greatly improved the paper. We are also grateful to Gregorio Curello, Piotr Dworczak,  Drew Fudenberg, Haoyu Liu, Paul Milgrom, Stephen Morris, Xiaosheng Mu, John Quah, Chris Shannon, Ludvig Sinander, Joel Sobel, and Bruno Strulovici, and participants of numerous seminars, for their comments. We thank Yutaro Akita, Nanami Aoi, C\'esar Barilla, Masato Eguchi, William Grimme, Taiyo Inamura, Takuto Kaneko, Haruki Kono, Tianhao Liu, Shinpei Noguchi, Ryo Shirakawa, Tomoya Takei, Kenji Utagawa, and Yutong Zhang for excellent research assistance. Yeon-Koo Che is supported by National Science Foundation Grant SES-1851821, and Fuhito Kojima is supported by the JSPS KAKENHI Grant-In-Aid 21H04979 and JST ERATO Grant Number JPMJER2301, Japan.  This work was supported by the Ministry of Education of the Republic of Korea and the National Research Foundation of Korea (NRF-2024S1A5A2A03038509).}}
\maketitle

\begin{abstract} The theory of Monotone Comparative Statics (MCS) has traditionally required a lattice structure,  excluding certain multi-dimensional environments like mixed-strategy games where this property fails. We show this structure is not essential. We introduce a weaker notion, pseudo lattice property, and preserve the theory's core results by generalizing the MCS theorems for individual choice and Tarski’s fixed-point theorem.  Our framework expands comparative statics to pseudo quasi-supermodular games. Crucially, it enables the first MCS analysis of mixed strategy Nash equilibria and (trembling-hand) perfect equilibria.

	\medskip
	
\end{abstract}


\section{Introduction}

Comparative statics analysis is a cornerstone of economic inquiry. It concerns how economic behavior—whether individual choices, collective outcomes, or game-theoretic equilibria—responds to changes in the underlying environment. As such, comparative statics provides the principal means by which economic models generate falsifiable predictions, forming a critical link between theory and empirical work. 

The modern methodological framework for this analysis is {\bf Monotone Comparative Statics (MCS)}. Pioneered by \cite{topkis:79, topkis:98}, \cite{milgrom/roberts:90}, and \cite{milgrom/shannon:94}, it is built upon a powerful order-theoretic foundation. The universal requirement of this framework is that the underlying domain of choice must be a {\bf lattice}---a partially ordered set where any two elements have a  least upper bound and  greatest lower bound.   This structure is essential for defining the notions of complementarity and supermodularity that are central to the theory, both with respect to parameters and among the dimensions of a choice variable.

The reliance on the lattice property, however, raises critical conceptual and practical questions. Conceptually, is this rigid mathematical structure truly necessary for monotone comparative statics? Practically, many important economic environments fail to satisfy the lattice property. A leading example is the domain of stochastic decisions.  The set of lotteries $\Delta(X)$ over a choice space, $X$, fails to form a lattice under standard stochastic orders, even when $X$ is itself a lattice. This limitation has severely constrained the application of MCS to settings involving risk, uncertainty, or strategic randomization.

This paper demonstrates that the analytical power of monotone comparative statics can be preserved without the rigid lattice requirement. We show that the lattice property can be replaced by a much weaker condition, a {\bf pseudo lattice}, which merely requires the existence of, possibly non-unique, {\it minimal} upper bounds and {\it maximal} lower bounds.  This property is remarkably general; for instance, any finite or compact set with largest and smallest elements is a {\bf complete pseudo lattice}.\footnote{While a formal definition is provided later, this is a mild requirement: for instance, any compact set that contains the largest and smallest elements is a complete pseudo lattice.} This relaxation  significantly expands the scope of MCS, allowing us to analyze environments previously beyond its reach, including mixed strategy Nash equilibria and, most strikingly, (trembling-hand) perfect equilibria.

Our analysis begins by rebuilding the theory for the canonical individual choice problem, substituting the standard lattice assumption with our weaker pseudo lattice structure. This section delivers two foundational results that parallel the cornerstones of the classic theory. First, we provide a full characterization—a necessary and sufficient condition—for the set of optimal choices by an individual to be monotone with respect to changes in the economic environment.  This result directly generalizes the celebrated Monotonicity Theorem of \cite{milgrom/shannon:94},  demonstrating that the core comparative statics conclusions can be obtained with virtually no loss of analytical power, with appropriately generalized ordinal conditions.  Second, we show that the set of maximizers inherits the complete pseudo lattice structure of the domain. This finding is analogous to the well-known result of \cite{milgrom/roberts:90} that the optimizer set forms a complete sublattice in the standard framework, a property that is often useful for equilibrium analysis. 

The analysis of equilibrium in strategic environments requires a tool for establishing existence and characterizing the structure of the solution set. To this end, we develop a  fixed-point theorem that serves as the analytical engine for our paper. This result generalizes the celebrated theorems by \cite{tarski:55} and \cite{zhou:94} by replacing the restrictive assumption of a complete lattice with that of a complete pseudo lattice. The theorem establishes that any monotonic correspondence from a complete pseudo lattice into itself has a nonempty set of fixed points that itself forms a complete pseudo lattice, thus admitting largest and smallest elements. We accompany this existence result with a comparative statics theorem (\Cref{thm:mcs of fp}) which shows that as the correspondence shifts upwards, the set of fixed points also shifts upwards in the weak-set order.

 With the results for individual choice and the fixed-point theorem in hand, we turn to the analysis of Nash equilibria in games with strategic complementarities. We apply our conditions to player payoffs to identify a broad class of pseudo quasi-supermodular games. In these games, the results from our individual choice analysis ensure that each player’s best-response correspondence is monotonic in the manner required by our generalized fixed-point theorem. Applying the theorem to the joint best-response correspondence immediately establishes our main equilibrium results: the set of pure-strategy Nash equilibria is nonempty and forms a complete pseudo lattice. Furthermore, when an exogenous parameter shifts the game in a way that strengthens each player’s incentive to take higher actions, the entire set of Nash equilibria shifts upward monotonically.

The power and scope of this generalized framework are best illustrated through its applications. We first analyze a generalized Bertrand game with product substitutes, showing that our machinery accommodates a wider range of competitive environments than the existing literature. Notably, our approach is capable of handling the pure Bertrand game, whose discontinuous payoff functions prevent it from being analyzed by standard supermodularity methods like those in \cite{milgrom/shannon:94}.  Second, and more fundamentally, we address the long-standing challenge of applying MCS to mixed strategy Nash equilibria. The domain of mixed strategies is a canonical example of a non-lattice space, making traditional methods inapplicable. We show that the set of mixed strategy equilibria is bounded by the extremal pure-strategy equilibria and therefore inherits their monotone comparative statics properties. While this result is implied by earlier work under the assumption of payoff continuity, our approach provides a more direct proof and, crucially, establishes the result without requiring payoff continuity. 

The ability of our framework to handle non-lattice domains, particularly the space of mixed strategies, culminates in our paper's most novel contribution: the first general monotone comparative statics analysis of (trembling-hand) perfect equilibria. Our analysis applies to a broad class of supermodular games. We extend \cite{selten:75}'s classic concept to potentially infinite games by defining a perfect equilibrium as the limit of Nash equilibria from a sequence of perturbed games. In each perturbed game, players are constrained to play from a parameterized set of full-support mixed strategies, ensuring that each open set of pure strategies  is played with some positive minimal  probability. 

Our proof strategy is to first establish existence and comparative statics properties for the Nash equilibria along the sequence of these perturbed games, and then to show that these properties are preserved in the limit. Specifically, we show that each perturbed game possesses extremal ``constrained-pure" Nash equilibria that bound the entire equilibrium set and shift monotonically as the environment changes. The limit of these extremal equilibria yields the existence of perfect equilibria in pure strategies and, crucially, ensures that the set of perfect equilibria inherits the monotone comparative statics property of equilibria of the perturbed games. This entire line of argument would be impossible with standard methods. The space of  mixed strategies in each perturbed game  forms a pseudo lattice but not a lattice, rendering traditional fixed-point theorems inapplicable. Thus, our generalized fixed-point theorem is the essential tool that enables the analysis at each step of the sequence.

\paragraph{Related Literature.} This paper contributes to the large and influential literature on the general methodology of monotone comparative statics. The workhorse methods in modern economic analysis were developed and refined in foundational contributions by   \cite{topkis:79, topkis:98}, \cite{vives:90}, \cite{milgrom/roberts:90}, \citet{milgrom/shannon:94}, and \citet{quah/strulovici:09}. A unifying feature of this entire body of work is its reliance on the mathematical structure of a lattice for the domain of choice, which provides the foundation for defining complementarities and comparing sets of optima. Our primary methodological contribution is to show that this structural assumption can be substantially weakened.
We demonstrate that the core analytical power of the theory is preserved when we dispense with the lattice property in favor of the more general pseudo lattice.

Several papers analyze  monotone comparative statics and fixed points of monotone operators in non-lattice environments.  \cite{quah2007comparative} considers individual choice problems in which a constraint set lacks a lattice structure.  Similarly, \cite{abian1961theorem}, \cite{smithson:71}, and \cite{li:14} establish fixed-point existence under monotonicity without lattice assumptions. Compared with these papers, we require a more structure on the domain but obtain stronger results, including the extremal predictions. More detailed comments will follow.

  The remainder of the paper is organized as follows. \Cref{sec:prel} introduces the preliminary mathematical concepts central to our analysis, including the formal definition of a pseudo lattice and the associated set orders. \Cref{sec:ind-choice} develops the monotone comparative statics results for the individual choice problem in this new domain. \Cref{sec:FP} presents
  our generalized fixed-point theorem to establish the existence, structure, and comparative statics of (pure and mixed) Nash equilibria. \Cref{sec:wmcs of thpe} extends the analysis to perfect equilibria. Finally, Section 6 offers concluding remarks. All proofs omitted from the main text are provided in the Appendix  and a Supplementary Appendix.

\section{Preliminaries} \label{sec:prel}

This section introduces a set of notions and terminologies and  establishes a number of preliminary results  that will be used  throughout the paper.  Our theory weakens the structural properties of the domain (e.g., choice or strategy sets for players) as well as the set order.    

\subsection{The structural properties of domain.}  Throughout, the choice domain $X$ is assumed to be  a  \emph{partially ordered set} with regard to a \emph{primitive partial order} $\ge$, namely a binary relation that is \emph{reflexive, transitive} and \emph{anti-symmetric} on $X$. For any set $S$, let $U_S := \{ x \in X : x \ge x',\forall x' \in S  \}$ and $L_S := \{ x \in X :  x \le x', \forall x' \in S\}$, that is, upper and lower contour sets of $S$, respectively. When $S= \{x\}$, we will simply write $U_x$ and $L_x$.  We assume that \(X\) is endowed with a topology under which the order \(\ge\) is \emph{closed}, meaning that the set $\{(x,y)\in X\times X : x \ge y\}$ is closed in the associated product topology.\footnote{This assumption is required to ensure that the induced stochastic order \(\ge^{sd}\) on \(\Delta(X)\) is a partial order---particularly, that it satisfies antisymmetry---when we analyze mixed-strategy Nash equilibria and perfect equilibria. In other parts of the analysis, specifically \Cref{lem:compact-chaincomp} and \Cref{cor:compact-with-ext}, it suffices to assume that \(\ge\) is \emph{semi-closed}, in the sense that   \(U_x\) and \(L_x\) are closed for each \(x\in X\).} Also, throughout the paper, we endow any space of probability measures with the weak topology.

Existing literature imposes additional order properties.  $X$ is a \textit{lattice} if for any $x,x'\in X$, $x\vee x'\in X$ and $x\wedge x' \in X$, where $x\vee x'  := \inf U_{\{x,x'\}}$ is their {\it join}, or the least (common) upper bound, of $\{x,x'\}$ and $x\wedge x':= \sup L_{\{x,x'\}}$ is their {\it meet}, or the greatest (common) lower bound, of $\{x,x'\} $.  (We will write $\vee_S$ and $\wedge_S$ when the sup or the inf is taken over a set $S \ne X$.)  $X$ is a  \emph{complete lattice} if,  for any $S\subset X$,  $\inf U_S \in X$ and   $\sup L_S \in X$, that is, its supremum and infimum exist in $X$.   A subset $S\subset X$ is a \emph{sublattice} of $X$  if, for any $x,x'\in S$, $x\vee x'\in S$ and $x\wedge x'\in S$.      A subset $S\subset X$ is a \emph{complete sublattice} of $X$ if $\inf U_{S'} \in S$ and   $\sup L_{S'} \in S$ for all $S' \subseteq S$.\footnote{Some other terminologies are used for the same notion: \cite{topkis:98} uses subcomplete sublattice and \cite{zhou:94} uses closed sublattice.  In particular, the ``closedness'' of \cite{zhou:94} should not be confused with the topological ``closedness'' used in this paper.}  

Throughout, we require much weaker structural properties for the partial order $(X,\ge)$. For any $x,x' \in X,$ let  $x \join x':= \{ y \in U_{\{x,x'\}}:z \in U_{\{x,x'\}} \Rightarrow y \not> z \}$  be their \textbf{pseudo join}---the set of minimal upper bounds of $\{x,x'\}$ in $X$---and let $x \meet x' := \{ y \in L_{\{x,x'\}}:  z \in L_{\{x,x' \}} \Rightarrow y \not< z \}$ be their \textbf{pseudo meet}---the set of maximal lower bounds of $x$ and $x'$ in $X$. If $X$ is a lattice, $x \join y$ and $x \meet y$ reduce to singleton sets $\{x \vee y\}$ and $\{x \wedge y\}$, respectively. In that sense, $\join$ and $\meet$ are generalizations of the usual join and meet operations to the non-lattice sets. We consider a non-lattice set where $\join$ and $\meet$ are well-defined.

We say that $X$ is a \textbf{pseudo lattice} if $x \join y$ and $x \meet y$ are nonempty for every $x,y \in X$.  We also say  $X$ is a \textbf{complete pseudo lattice} if it is  chain complete and, for all $S\subset X$, both 
 $\Join_X S:=\{ z\in U_S: x\in U_S \Rightarrow  x\not < z\}$ and  $\Meet_X S:=\{ z\in L_S: x\in L_S \Rightarrow  x\not > z\}$ are nonempty.\footnote{A partially ordered $X$ is \emph{chain complete} if every chain  $C \subset X$ has a supremum  and infimum in $X$. (This notion is sometimes called ``chain complete in both directions.'')  Note that a chain is a totally ordered subset $C\subset X$; that is, for any $x,y \in C$, either 
$x \le y $ or $y\le x $.} 
  A subset $S\subset X$ is a \textbf{weak pseudo sublattice} of $X$  if, for any $x,x'\in S$, $(x\join x')\cap S$ and $(x\meet x')\cap S$ are both nonempty, and a {\bf pseudo sublattice} if, for any $x,x'\in S$, $x\join x' \subset S$ and $x\meet x' \subset S$.   Clearly, if $S$ is a  pseudo sublattice of $X$, then it is a weak pseudo sublattice of $X$; but the converse need not hold.
   A subset $S$ of $X$ is a \textbf{complete pseudo sublattice} if it is chain complete and, for every nonempty $S' \subseteq S$, $\Join_X S'$ and $\Meet_X S'$ are nonempty subsets of $S$.

A pseudo lattice and a complete pseudo lattice are considerably weaker than a lattice and a complete lattice. As we will see later (\Cref{cor:compact-with-ext}), any compact set is a complete pseudo lattice if (and only if) it contains the largest and smallest elements.  We provide  several examples of $X$ that is a pseudo lattice but not a lattice. 

\begin{example}\label{example2} Every finite set $X$ is compact, and hence it becomes a complete pseudo lattice whenever it contains both the largest and smallest elements.  To see such a set need not be a lattice,  consider $X=\{ (0,0), (1,0), (0,1), (2,1), (1,2), (3,3)    \}$.
\begin{figure}[h] \label{fig:finite pseudo lattice}
	\centering
    \begin{tikzpicture}[scale=1.2]

  \draw[->, line width=1.2pt] (-0.5,0) -- (3.5,0);
  \draw[->,line width=1.2pt] (0,-0.5) -- (0,3.5);

  \foreach \x in {0,1,2,3}
    \draw[gray!80, dotted, line width=0.8pt] (\x,-0.5) -- (\x,3.5);
  \foreach \y in {0,1,2,3}
    \draw[gray!80, dotted, line width=0.8pt] (-0.5,\y) -- (3.5,\y);

  \foreach \point in {(0,0),  (1,0),(0,1), (2,1), (1,2), (3,3)} {
    \fill \point circle (2pt);
  }

  \node[below left] at (0,0) {$(0,0)$};
  \node[below] at (1,0) {$x=(1,0)$};
  \node[left] at (0,1) {$x'=(0,1)$};
  \node[below right] at (2,1) {$(2,1)$};
  \node[above left] at (1,2) {$(1,2)$};
  \node[above right] at (3,3) {$(3,3)$};

\end{tikzpicture}\caption{A complete pseudo lattice that is not a lattice.} \end{figure} 

This set has  the largest and smallest elements, $(3,3)$ and $(0,0)$, and thus forms a complete pseudo lattice. However, it is not a lattice: for example, for $x =(1,0)$ and $x' = (0,1)$,
$x \join x' =\{(2,1), (1,2)\}$. 
Since this set is non-singleton,   $(1,0)\vee (0,1)$ is not well-defined, so $X$ fails to be a lattice.
\end{example}

\begin{example}[\textbf{Stochastic dominance order}]  \rm  \label{KKO example}
Consider the set $X=\Delta(S)$ of all Borel probability measures on a compact, partially ordered Polish space $S$ containing largest and smallest elements $\overline s$ and  $\underline s$. We endow $X$ with the (first-order) stochastic dominance order, $\ge^{sd}$: $x \ge^{sd} y$ if $\int f  dx \ge \int f dy$ for all bounded nondecreasing  function $f: S \to \mathbb{R}$.\footnote{This is equivalent to requiring that $x (S') \ge y (S')$ for every upward closed set $S' \subset S$, where $S'$ is upward closed if $s \in S'$ and $s' \ge s$ imply $s' \in S'$. That $\ge^{sd}$ is a partial order is shown by \citet{kamae1978stochastic}.} If $S$ contains largest and smallest elements ($\overline{s}$ and $\underline{s}$), then $X$ is a complete pseudo lattice by \Cref{cor:compact-with-ext}, as it is compact and has largest ($\delta_{\overline{s}}$) and smallest ($\delta_{\underline{s}}$) elements (where $\delta_s$ is the Dirac measure at $s$).

However, $X$ is not a lattice in general, particularly if $S$ is multidimensional, as an example from \cite{kamae1977stochastic} shows. 
Let $S=\{0,1\}\times\{0,1\}$, as depicted below.
\begin{figure}[h]
	\centering
	\begin{tikzpicture}[scale=3.5]
		\fill (0,0) circle (1pt) node[below left=1pt] {$\frac{1}{2}$};
		\fill (1,0) circle (1pt) node[below right =2pt] {$\frac{1}{2}$};
		\fill (0,1) circle (1pt) node[above left =1pt] {$\frac{1}{2}$};
		\fill (1,1) circle (1pt) node[above right=2pt] {$\frac{1}{2}$};
		
		\draw[blue, thick,rounded corners, rotate around={45:(0.5,0.5)}] (-0.5,0.3) rectangle ++(2.0,0.4);
		
		\draw[red, thick,rounded corners, rotate around={135:(0.5,0.5)}] (-0.5,0.3) rectangle ++(2.0,0.4);
		
		\draw[thick, rounded corners] (-0.25,-0.27) rectangle ++(1.53,0.4);
		
		\draw[thick, rounded corners] (-0.25,-0.27) rectangle ++(0.4,1.53);
		
	\end{tikzpicture}
	\caption{Failure of the lattice property} \label{fig:KKO}
\end{figure}

To see why $X=\Delta(S)$ fails to be a lattice, consider the two lotteries $x=\frac{1}{2}\delta_{(1,1)}+\frac{1}{2}\delta_{(0,0)}$ and $x^{\prime}=\frac{1}{2}\delta_{(1,0)}+\frac{1}{2}\delta_{(0,1)}$. Both lotteries $\underline{a} = \frac{1}{2}\delta_{(0,0)}+\frac{1}{2}\delta_{(1,0)}$ and $\underline{b} = \frac{1}{2}\delta_{(0,0)}+\frac{1}{2}\delta_{(0,1)}$ are maximal lower bounds of $x$ and $x'$. Since  $x \meet x'$ is a non-singleton set, $X$ fails to be a lattice. This space is a canonical example for our later analysis of mixed strategies.
\end{example}

\begin{example}[\textbf{Mean-preserving spread/convex order}]
Another important example is the set of information structures (distributions of posteriors) ordered by the mean-preserving spread, or convex, order ($\ge^{cx}$). Consider a compact, convex set $\Theta \subset \mathbb{R}^d$. For a fixed prior $\mu \in \Delta(\Theta)$, the set $X_\mu$ of all  distributions $x \in \Delta (\Delta (\Theta))$  over posterior beliefs with mean $\mu$ constitutes a feasible set.\footnote{The order  $\ge^{cx}$  is  defined as: $x \ge^{cx}y$  if $\int f dx \ge \int f dy$ for all continuous, convex function $f: \Delta(\Theta) \to \mathbb{R}$. This defines  a partial order. In particular, Corollary 3.24 of \citet{elton1992fusions} ensures the antisymmetry of $\ge^{cx}$ since the set is a compact metrizable convex subset of a locally convex topological vector space.} Note that this is a compact set and contains a largest and a smallest element,  $\overline x_\mu$ and $\underline x_\mu$:   $\underline{x}_\mu$  represents the no-information experiment; $\overline{x}_{\mu}$ represents the experiment that fully reveals every $\theta \in \Theta$.
Hence, by \Cref{cor:compact-with-ext}, $X_{\mu}$ is a complete pseudo lattice.\footnote{Often, we restrict attention to the distribution over posterior \emph{means} rather than posterior \emph{beliefs}.  The feasible set of distributions of the posterior mean then becomes the convex-order interval 
$\{\, \nu \in \Delta(\mathbb{R}^d) : \delta_{\bar{m}} \le^{cx} \nu \le^{cx} \mu \,\}$, 
where $\bar{m} = \int_\Theta \theta\, d\mu(\theta)$. This set is also a complete pseudo lattice.} However,  $X_{\mu}$ generally fails to be a lattice as shown in \Cref{exa:failure_lattice} of the Online Appendix.
\end{example}

The next result characterizes a complete pseudo lattice \(X\) in terms of the existence of its extremal elements.  
In particular, \(X\) is a complete pseudo lattice whenever it is compact (or, more generally, chain complete) and possesses the largest and smallest elements.  

\begin{theorem}\label{thm:pseudo-lattice-chain-complete}
A chain complete set \(X\) is a complete pseudo lattice if and only if it admits both the largest and smallest elements.
\end{theorem}

Chain completeness is a mild condition.
It is automatically satisfied under compactness, which is often assumed for other purposes.

\begin{lemma}\label{lem:compact-chaincomp}
If \(X\) is compact, then it is chain complete.\footnote{The converse does not hold.   For example, \(X=[0,1)\cup\{2\}\) is chain complete but not compact.} 
\end{lemma}

\begin{corollary}\label{cor:compact-with-ext}
A compact set $X$  is a complete pseudo lattice  if and only if  it admits both a largest and a smallest element.
\end{corollary}




\subsection{Set orders} 

Consider a nonempty pseudo lattice $(X,\ge)$.  One can define several set orders induced by $\ge$, and two of them are of particular interest to us: \emph{pseudo strong-set (pSS) order} $\ge_{pss}$ and \emph{weak-set (WS) order} $\ge_{ws}$.

We say $S'\subset X$ {\bf pSS dominates} $S\subset X$, and write $S'\ge_{pss} S$, if $\forall x\in S, \forall x'\in S'$, $x\join x'\subset S'$ and   $x\meet x'\subset S$.  $S'\subset X$ {\bf weak pSS dominates} $S\subset X$, and write $S'\ge_{wpss} S$, if $\forall x\in S, \forall x'\in S'$, $(x\join x')\cap S'$ and $(x\meet x')\cap S$ are both nonempty. We say  $S'$ {\bf weak-set dominates} $S$, and write $S'\ge_{ws} S$, if, for each $x\in S$, there exists $x'\in S'$ such that $x'\ge x$; and for each $x'\in S'$, there exists $x\in S$ such that $x\le x'$.

As  is easily seen, $S' \ge_{pss} S$ implies $S'\ge_{ws} S$.  The following result further clarifies  the relationship between these two orders by decomposing pseudo strong-set order  into   weak-set order and a couple of ``extra properties'' when the choice domain is a pseudo lattice (and the compared sets are pseudo sublattices):\footnote{One can easily construct examples showing that each property is  indispensable for this characterization.}  
\begin{theorem} \label{thm:ss-char}
 Consider  a pseudo lattice $X$ and its subsets $S$ and $S'$. Then,  $S' \ge_{pss} S $  if (i)  $S' \ge_{ws}  S$; (ii) $S \cup S'$ is a pseudo sublattice;  and  (iii) (sandwich property)  for  any $x \in S$  and $y, z \in S'$  (resp., any $x \in S'$ and $y, z \in S$),  $x \in   [y,z]$ implies $x \in S'$ (resp., $x \in S$). 
  Conversely, if $S$ and $S'$ are nonempty pseudo sublattices, then   $S' \ge_{pss} S$ implies  the properties   (i)--(iii).  
\end{theorem}
\begin{proof}
To prove the first statement, let us consider any $x \in S$ and $x' \in S'$.   To show $x \join x' \subset S'$, suppose not. Then, there exists  $\hat  x\in x \join x'$ such that \jw{$\hat x \not\in S'$} and $\hat x\in S$ by (ii). By (i), there exists $z \in S'$ such that $z  \ge \hat  x$. So we have $x' \le \hat x \le z$ while $x' ,z \in S'$ and $\hat x \in S$. Thus, by (iii), $\hat  x \in S'$, a contradiction.  The proof of $x \meet x' \subset S$ is analogous and  omitted. 

Suppose now that $S' \ge_{pss} S$ where $S$ and $S'$ are nonempty pseudo sublattices. 
Clearly, (i) holds. To see that (ii)  holds, consider any $x, x' \in S \cup S'$. If either $x, x' \in S$ or $x,x' \in S'$, then  clearly $x \join x'$ and $x\meet x'$ are subsets of $S \cup S'$ since $S$ and $S'$ are  pseudo sublattices. If $x \in S$ and $x' \in S'$, then $S' \ge_{pss} S$ implies that both $x \join x' $ and $x \meet x'$ are subsets of  $S \cup S'$.  To verify (iii), observe that for any $x \in S$ and $y,z \in S'$ with  $  x  \in [y, z]$,  we have  $\{x\} =   x \join y  $ and thus $x \in S'$ since $S' \ge_{pss} S$. Also,  for any $x \in S'$ and $y,z \in S$ with  $ x \in [y,z]$,  we have  $\{x\} =   x \meet z $ and thus $x \in S$.    
\end{proof}

This characterization clarifies precisely what is “lost’’ when we use the weak-set order instead of the (pseudo) strong-set order---namely, properties (ii) and (iii).


\section{Individual Choices}\label{sec:ind-choice}

Consider an individual who chooses an action $x$ from a feasible set $S \subset X$, facing a parameter $t \in T$, by maximizing an objective function $u : X \times T \to \mathbb{R}$, where $X$ is a pseudo lattice and $T$ is a partially ordered set.  We are interested in how the set of optimal choices,
$$M_{S} (t) := \arg\max_{x \in  S} u (x,t), $$
 responds to changes in  the {\it environment} from $(t,S)$ to $(t',S')$. We write $M(t)$ when $S=X$ and $M_S$ when $t$ is fixed. The next subsection provides conditions that characterize the monotone comparative statics for the set $M_S(t)$ in the weak pseudo strong-set order.

\subsection{MCS characterization of individual choices}

We begin with some definitions. First, we say that $f: X\to \mathbb{R}$ is  \textbf{pseudo quasi-supermodular} if, for any $x, y\in X$, $\overline x\in x\join y$, and $\underline x\in x\meet y$,
\begin{align*}
    f(x)-f(\underline x) \ge\mathrel{(>)}0 \Rightarrow f(\overline x)-f(y)\ge\mathrel{(>)}0, \tag{pQSUP} \label{pqsup}
\end{align*}
 and  \textbf{pseudo  supermodular} if, for any $x, y\in X$, $\overline x\in x\join y$, and $\underline x\in x\meet y$,
\begin{align*}
  f(x)-f(\underline x)\le f(\overline x)-f(y). \tag{pSUP} \label{psup}
\end{align*}
As is well-known, the prefix ``quasi'' weakens the notion from a cardinal concept to an ordinal one.  More importantly  for our purposes, ``pseudo'' extends the property to a pseudo lattice.  If $X$ is a lattice (i.e., not just a pseudo lattice), then  pseudo quasi-supermodularity reduces to {\it quasi-supermodularity}, and pseudo supermodularity reduces to {\it supermodularity}, well known in the literature.  These conditions capture the sense of complementarities across multiple dimensions of action under the payoff function $f$.

Next, let us consider a parametrized family of functions $u: X \times T  \to \mathbb R$.  
 We say that $u(x,t)$ satisfies \emph{single-crossing  in $(x,t)$} if, for each $x,x' \in X$ with $x' \ge x$ and $t, t' \in T$ with $t' \ge t$, 
\begin{align*}
    u(x',t)-u(x,t)\ge\mathrel{(>)}0 \Rightarrow u(x',t')-u(x,t') \ge\mathrel{(>)}0, \tag{SC} \label{SC}
\end{align*}
 and \emph{increasing differences  in $(x,t)$}  if, for each $x,x' \in X$ with $x' \ge x$ and $t, t' \in T$ with $t' \ge t$, 
\begin{align*}
    u(x',t)-u(x,t) \le u(x',t')-u(x,t'). \tag{ID} \label{ID}
\end{align*}
These conditions capture the idea that an individual faces a stronger incentive to raise her action under $u(\cdot,t')$ than under $u(\cdot,t)$.  As is clear, increasing differences imply single-crossing. Also, we say that $u$ is pseudo (quasi-)supermodular in $x$ if $u(x, t)$ is pseudo (quasi-)supermodular as a function of $x$ for each $t \in T$. 

 \cite{milgrom/shannon:94} characterize monotone comparative statics of the optima in the strong-set order for lattice domains with quasi-supermodular and single-crossing payoffs.  We extend the characterization when $X$ is (only) a pseudo lattice.\footnote{\cite{quah2007comparative} considers the individual choice problem in which the constraint set (defined in a lattice) lacks a sublattice property and identifies conditions under which optimal choices rise in weak-set order when a constraint set also rises in weak-set order. The current theorem strengthens both the hypothesis (a monotonic pSS shift of the constraint set) and the conclusion (a monotonic pSS shift of the optimal choices), and hence complements \cite{quah2007comparative}.}  



\begin{theorem} [Characterization] \label{thm:MCS-PSS}  $M_{S'} (t') \ge_{wpss} M_{S} (t)$ for all $t' \ge t$ and   $S'\ge_{wpss} S$ if and only if $u$ is pseudo quasi-supermodular in $x$ and satisfies single-crossing in $(x,t)$.
\end{theorem}

This result reduces to the aforementioned result by \citet{milgrom/shannon:94}, their Theorem 4, if $X$ is a lattice, since  the qualifier  ``pseudo'' then becomes immaterial. The pseudo quasi-supermodularity and single-crossing conditions, just like quasi-supermodularity and single-crossing under the lattice environment, play a key role in our subsequent applications. These conditions deliver MCS in the ``weak'' pSS order, but they require a correspondingly weaker condition for the feasible sets.  We can obtain MCS in the ``strong'' pSS order, for free, under a stronger condition on the feasible sets:

\begin{proposition} [Monotonicity] \label{cor:MCS-sublat}    If $u$ is pseudo quasi-supermodular in $x$ and satisfies single-crossing in $(x,t)$, then 
$M_{S'} (t') \ge_{pss} M_{S} (t)$  for all $t' \ge t$ and $S'\ge_{pss} S$.
\end{proposition}
\begin{proof}   For any $x \in M_S(t)$ and $x' \in M_{S'}(t')$, choose $z \in x \meet x'$ and $z' \in x \join x'$.  Then, $z\in S$ and $z'\in S'$ (from $S'\ge_{pss} S$), so  $u(x,t) \ge u(z,t)$ and, by pseudo quasi-supermodularity, it follows that $u(z',t) \ge u(x',t)$.  By single-crossing, $u(z',t') \ge u(x',t')$, so $z' \in M_{S'}(t')$. That $z \in M_S(t)$ follows analogously from the strict inequality parts of \eqref{pqsup} and \eqref{SC}.
\end{proof}

Since the cardinal conditions imply the ordinal ones, the following is trivial:
\begin{corollary} \label{cor:MCS-cardinal}   If $u$ is pseudo supermodular in $x$ and satisfies increasing differences in $(x,t)$, then 
$M_{S'} (t') \ge_{pss} M_{S} (t)$  for all $t' \ge t$ and $S'\ge_{pss} S$.
\end{corollary}

As in \cite{milgrom/shannon:94}, a strengthening of single crossing yields monotonicity regardless of how an optimal choice is selected.

\begin{corollary} [Monotone Selection] \label{cor:MCS-sscp}  Suppose $u$ is pseudo quasi-supermodular in $x$ and satisfies strict single-crossing in $(x,t)$: for any $x'>x$, $t' > t$,
$$u(x',t)-u(x,t)\ge 0 \Rightarrow u(x',t')-u(x,t')  >0,$$ and assume $S'\ge_{pss} S$.  Then,  $\forall  z\in M_{S} (t)$,
$\forall z'\in M_{S'} (t')$, we have $z'\ge z.$
\end{corollary}

\begin{proof} Fix $z \in M_S(t)$ and $z' \in M_{S'}(t')$. 
By \Cref{cor:MCS-sublat}, any $\tilde z \in z \meet z'$ satisfies $\tilde z \in M_S(t)$, so $u(z,t) - u(\tilde z,t) = 0$. 
If $z' \not\ge z$, then $\tilde z < z$, and strict single crossing implies $u(z,t') - u(\tilde z,t') > 0$. 
Pseudo quasi-supermodularity then implies $u(\hat z,t') - u(z',t') > 0$ for some $\hat z \in z \join z' \subset M_{S'}(t')$, contradicting $z' \in M_{S'}(t')$.  
\end{proof}

\paragraph{Directional weakening.} 
One can provide a weaker ``directional'' version of \Cref{thm:MCS-PSS} by splitting the conditions into upper and lower versions. We say  $S'\subset X$ \textbf{upper} (resp. \textbf{lower}) \textbf{weak pSS dominates} $S\subset X$, and write $S'\ge_{uwpss} S$ (resp. $S'\ge_{lwpss} S$), if $\forall x\in S, \forall x'\in S'$, $x\join x'\cap S'$ (resp.  $x\meet x' \cap S$) is nonempty.  Next, we say that a function $u$ is \textbf{upper} (resp. \textbf{lower}) \textbf{pseudo quasi-supermodular} if the weak (resp. strict) inequality part of \eqref{pqsup} holds, and that $u$ is upper (resp. lower) single-crossing if the weak (resp. strict) inequality part of \eqref{SC} holds.
  Then, we have:
\begin{manualtheorem}{3$'$} \label{thm:MCS-UPSS}
$M_{S'} (t') \ge_{uwpss} M_{S} (t)$ for all $t' \ge t$ and   $S'\ge_{wpss} S$ if and only if $u$ is upper pseudo quasi-supermodular in $x$ and satisfies upper  single-crossing in $(x,t)$.  An analogous statement holds for the lower case.\footnote{That is, the same statement holds when $\ge_{uwpss}$, upper pseudo quasi-supermodularity, and upper single-crossing are replaced by  $\ge_{lwpss}$, lower pseudo quasi-supermodularity, and lower single-crossing, respectively.}   
   \end{manualtheorem}
The upper and lower versions of \Cref{thm:MCS-UPSS} jointly imply \Cref{thm:MCS-PSS}, which is how the latter theorem is proven in \Cref{app:ind-choice}.
 
\subsection{The structure of the maximizer set}

As in the lattice case, the complementarity property of the payoff function over a pseudo lattice domain induces a distinctive structure for the set of maximizers. To ensure that this set is well defined, we assume that the payoff function $u : X \to \mathbb{R}$ is \textbf{order upper semicontinuous}: for any chain $C = \{ {x_\alpha} \}_\alpha \subset X$, 
 $$\inf_{\alpha }  \sup_{ \beta \ge \alpha } u (x_\beta ) \le u (x), $$ where $x$ is the  infimum or supremum of $C$, if it exists.  

\begin{theorem} \label{thm:pseudo-argmax} Assume $X$ is a pseudo lattice and $u: X \to \mathbb R$ is   pseudo quasi-supermodular. 
\begin{description}
    \item[(i)] $\arg\max_{x\in X} u (x)$ is  a pseudo sublattice of $X$ whenever it is nonempty.

    \item [(ii)]  If, in addition, $X$ is a complete pseudo lattice and $u$ is order upper semicontinuous, then $\arg\max_{x\in X} u (x)$ is a nonempty, complete  pseudo sublattice, admitting the largest and the  smallest  point.  
\end{description}  
\end{theorem}

That the maximizers inherit the complete pseudo lattice structure is interesting in its own right, but it will also be crucial for the subsequent analyses of games.  This result generalizes Theorem 2 of \cite{milgrom/roberts:90}  and Theorem A4 of \cite{milgrom/shannon:94}. However, the lack of a lattice structure of the domain makes the argument distinct and more involved.

\subsection{Decision under uncertainty} \label{sec:decision under uncertainty} 

We now apply our framework to optimal behavior under uncertainty---a setting central to mixed-strategy equilibrium analysis---where an agent chooses a potentially random action to maximize expected utility. Let $u: A \times \Theta \to \mathbb{R}$ be the utility function, with $A$ and $\Theta$ denoting partially ordered sets of actions and states, respectively. The agent's uncertainty is represented by a probability distribution $\eta \in \Delta(\Theta)$ over the state space; in game-theoretic contexts, $\eta$ often represents the mixed strategy profile of opponents. For any randomized action $x \in X = \Delta(A)$ and state distribution $\eta \in \Delta(\Theta)$, we define the expected utility  as
\begin{equation} \label{expected-utility}
\overline{u}(x, \eta) := \int_A \left( \int_{\Theta} u(a, \theta) \eta(d\theta) \right) x(da) 
\end{equation} and the associated choice correspondence as   \begin{equation}
    \label{expected-choice}    M(\eta) := \arg \max_{x \in X} \overline{u}(x, \eta).
\end{equation}  
 We endow both $\Delta(A)$ and $\Delta(\Theta)$ with the stochastic dominance order induced by the underlying orders on $A$ and $\Theta$. 

Our aim is to establish the monotone comparative statics of $M(\eta)$ as $\eta$ shifts.  However, applying \Cref{thm:MCS-PSS} directly to this lottery space is difficult. As illustrated next, the induced expected utility function $\overline{u}(x, \eta)$ may fail to satisfy pseudo quasi-supermodularity  even when the underlying utility $u(a, \theta)$ is supermodular.

\begin{example}[\textbf{Failure of pQSUP}] \label{ex:expected-utility-fail}
Consider the space described in \Cref{KKO example}, where $X = \Delta(\{0,1\} \times \{0,1\})$ is endowed with the stochastic dominance order, alongside a utility function $u: \{0,1\} \times \{0,1\} \to \mathbb{R}$ defined by $u(0,0)=u(1,0)=2$, $u(0,1)=0$, and $u(1,1)=1$. While $u$ is supermodular, the associated expected utility $\overline{u}$ fails \eqref{pqsup}. To see this, let $x=\frac{1}{2}\delta_{(0,0)}+\frac{1}{2}\delta_{(1,1)}$, $x^{\prime}=\frac{1}{2}\delta_{(1,0)}+\frac{1}{2}\delta_{(0,1)}$, $\underline{x}=\frac{1}{2}\delta_{(0,0)}+\frac{1}{2}\delta_{(0,1)}$, and $\overline{x}=\frac{1}{2}\delta_{(0,1)}+\frac{1}{2}\delta_{(1,1)}$. Note that $\underline{x} \in x \meet x'$ and $\overline{x} \in x \join x'$. We have $\overline{u}(x) = 3/2 > 1 = \overline{u}(\underline{x})$, but $\overline{u}(\overline{x}) = 1/2 < 1 = \overline{u}(x')$, violating \eqref{pqsup}.
\end{example}

Since  pseudo quasi-supermodularity is necessary for the characterization in \Cref{thm:MCS-PSS}, its failure implies that optimal choices do not generally exhibit monotone comparative statics on arbitrary subsets $S \subset \Delta(A)$, such as the set $\{x, x', \underline{x}, \overline{x}\}$ in the example. Nevertheless, we can recover monotone comparative statics in ``standard'' cases where the choice set includes pure actions.  Our analytical approach is to exploit the linearity of the expected payoff $\overline{u}$ in $x \in X$. In such cases, monotone comparative statics of $M(\eta)$ can be obtained in the weak-set order directly from standard conditions on the underlying utility $u$.

\begin{proposition} \label{cor:random_choice}
Assume that $A$ is a complete pseudo lattice.  Assume also that  $u(a,\theta)$ is bounded, order upper semicontinuous,  pseudo supermodular in $a$, and satisfies increasing differences in $(a,\theta)$. Then, for each $\eta \in \Delta (\Theta)$, $M(\eta):=\arg\max_{x\in \Delta(A)} \overline{u} (x,\eta)$ has largest and smallest elements, both of which belong to $A$ and are nondecreasing in $\eta$. Hence, $M(\eta')\ge_{ws} M(\eta)$ for all $\eta,\eta' \in \Delta (\Theta) $ with $\eta ' \ge^{sd} \eta$.
\end{proposition} 

The result consists of two main observations. Consider first a maximization where the choice set is restricted to only pure actions $A$. By \Cref{thm:pseudo-argmax}, this restricted problem admits extremal optima, and \Cref{thm:MCS-PSS} implies that they are pSS monotonic in $\eta$  with respect to first-order stochastic dominance. This is because pseudo supermodularity and increasing-difference---the cardinal conditions---are preserved under expectation {\it over pure actions}.

Now, when we extend the choice space to the lottery space $\Delta(A)$, these pure actions remain optimal because $\overline{u}(\cdot,\eta)$ is linear in lotteries. Moreover, they remain extremal because any optimal lottery must assign its entire probability mass within the interval bounded by the two extremal pure optima. Note, however, that we still pay a price for the failure of pseudo quasi-supermodularity: the monotone comparative statics property holds only in the weak-set order. As will be seen in the following sections, \Cref{cor:random_choice} is sufficiently powerful to establish the desired comparative statics for mixed-strategy Nash equilibria and perfect equilibria.

\section{Fixed Points and Nash Equilibria}   \label{sec:FP}

This section develops the analytical machinery for equilibrium analysis on non-lattice domains. We first establish a fixed-point theorem for pseudo monotonic correspondences along with associated comparative statics results. We then apply these tools to generalize the theory of games with strategic complementarities, providing a unified framework that encompasses both pure and mixed strategies.
 

 \subsection{Existence and comparative statics of fixed points} \label{sec:fixed point}


Consider a self correspondence $F: X \rightrightarrows X$ defined over a complete pseudo lattice $X$ endowed with a partial order $\ge$.  We say that an element $x\in X$ is a \emph{fixed point}  of  $F$  if $x\in F(x)$.   

Let us call a self-correspondence $F: X \rightrightarrows X$ \textbf{pseudo monotonic} if (i) $F(x)$ is a nonempty complete pseudo sublattice for each $x \in X$ and (ii) $F$ is pSS monotonic, i.e.,  $F(x') \ge_{pss} F (x) $ for all $x', x \in X$ with $x' \ge x$.

\begin{theorem} \label{thm:tarski-correspondence}
If $F: X \rightrightarrows X$ is a pseudo monotonic correspondence on a complete pseudo lattice $X$, then its fixed-point set is a nonempty complete pseudo lattice, thus admitting the largest and smallest points.
\end{theorem}

When $F$ is a function, the following generalization of \cite{tarski:55}'s theorem obtains: 

\begin{corollary} \label{cor:tarski} If $F:X\to X$ is a nondecreasing function on a complete pseudo lattice $X$, then its fixed-point set  is a nonempty complete pseudo lattice. 
\end{corollary}

The theorems by \cite{zhou:94} and  \cite{tarski:55} are the closest antecedents of \Cref{thm:tarski-correspondence} and \Cref{cor:tarski}, respectively.  The crucial differences are that they require $X$ to be a complete lattice and  $F$, in case of a correspondence, to satisfy additional lattice properties.\footnote{Specifically, the correspondence $F$ is required to be complete sublattice-valued and monotonic in the strong-set dominance sense; see \cite{zhou:94}.}  Our results relax these lattice properties to pseudo-lattice counterparts. 
Recall these properties are substantially weaker.  For example, a set $X$ is a complete pseudo lattice if $X$ is compact and contains the largest and smallest points (\Cref{cor:compact-with-ext}).\footnote{\cite{abian1961theorem}, \cite{smithson:71}, and \cite{li:14}'s fixed-point theorems require even weaker conditions. For instance, the version of Li's theorem invoked by  \cite{che2019weak}  requires $X$ to be only a compact partially-ordered set, together with some regularity conditions ensuring the existence of an upper diagonal or lower diagonal point.   The complete pseudo lattice condition is not much stronger than this; it is weaker than the compactness with extremal points (see \Cref{lem:compact-chaincomp}), which only strengthens the regularity condition.  However, the results are considerably more powerful: fixed points are a complete pseudo lattice, so they contain the largest and smallest points, which is not guaranteed by that theorem (see \cite{che2019weak}).  As will be seen, these properties will be used crucially for the later application, particularly the monotone comparative statics of perfect equilibria.} The only price paid is that the fixed points form a complete pseudo lattice rather than a complete lattice.  Crucially, a complete pseudo lattice always admits extremal elements; thus, our framework remains sufficiently powerful to guarantee the existence of largest and smallest fixed points, which provides the foundation for the equilibrium analysis in the subsequent sections.

Crucially for our purposes, pseudo monotonic correspondences are readily amenable to monotone comparative statics analysis. Let   $\Fp(t)$ denote the set of fixed points of a parametrized self-correspondence  $F (\cdot,t): X \rightrightarrows  X$.  

\begin{theorem} \label{thm:mcs of fp}  For a family of  pseudo monotonic self-correspondences $F (\cdot, t)$  on a complete pseudo lattice $X$,  if $F (x,t') \ge_{ws} F (x,t)$ for all $x \in X$,  then  $\Fp (t') \ge_{ws} \Fp (t)$.
\end{theorem}

\paragraph{Directional weakening.}
As in the individual-choice setup, one may introduce directional weakenings of the monotonicity requirement for the fixed-point correspondence.\footnote{See the generalized Bertrand game in \Cref{sec:pQSup_games} for an instance where such a weakening is needed.}  A subset $S$ of $X$ is a \textbf{complete upper} (resp. \textbf{lower}) \textbf{pseudo sublattice} if it is chain complete and, for every nonempty $S' \subseteq S$, $\Join_X S'$ (resp. $\Meet_X S'$) is a  nonempty subset of $S$. We say  $S'\subset X$ \textbf{upper} (resp. \textbf{lower}) \textbf{pSS dominates} $S\subset X$, and write $S'\ge_{upss} S$ (resp. $S'\ge_{lpss} S$), if $\forall x\in S, \forall x'\in S'$, $x\join x'\subset S'$ (resp. $x\meet x'\subset S$).

If $F(x)$ is a nonempty complete upper  pseudo sublattice for each $x \in X$, and
$F (x') \ge_{upss}  F (x), \forall  x' \ge x$, then  we say $F$ is \textbf{upper pseudo monotonic}---and \textbf{lower pseudo monotonicity} is defined analogously.   \begin{manualtheorem}{5$'$} \label{thm:5prime}
If $F: X \rightrightarrows X$ is an upper (resp. lower) pseudo monotonic correspondence on a complete pseudo lattice $X$, then  its fixed-point set   is  nonempty and admits the largest (resp. smallest) point.
 \end{manualtheorem} 
To establish the comparative statics result, let us weaken the weak-set order similarly: we say that $S'$ \textbf{upper weak-set dominates} $S$, and write $S'\ge_{uws} S$, if  for every $x\in S$, there exists $x'\in S'$ with $x'\ge x$;  $S'$ \textbf{lower weak-set dominates} $S$, written $S' \ge_{lws} S$, if  for every $x'\in S'$, there exists $x\in S$ with $x\le x'$.
  \begin{manualtheorem}{6$'$} \label{thm:6prime}
Let  $F (\cdot, t)$ and $F (\cdot,t')$ be  self-correpondences defined on a complete pseudo lattice $X$. If $F (\cdot, t')$ is upper pseudo monotonic and $F (x, t') \ge_{uws}  F(x, t)$ for all $ x \in X$, then $\Fp (t') \ge_{uws} \Fp (t)$.\footnote{Notice that the monotonicity  restriction is imposed only on 
$F(\cdot, t')$, but \emph{not} on $F ( \cdot,t)$. This generality will play an important role in our later analysis.} If $F (\cdot, t)$ is lower pseudo monotonic and $F (x, t') \ge_{lws} F (x, t) $ for all $ x \in X $, then $ \Fp (t') \ge_{lws} \Fp (t)$.\footnote{This result is related to Theorem 3 in \cite{acemoglu2015robust}, in that both the fixed-point operator and the fixed-point set shift in the weak-set order. Their analysis, however, does not rely on pseudo-monotonicity; instead, it requires $X$ to be compact and $F$ to be upper hemicontinuous.}
 \end{manualtheorem}

\subsection{Pseudo quasi-supermodular games} 
\label{sec:pQSup_games}

Consider a normal-form game $\Gamma=(I,S,u)$, where $S=\times_{i\in I}S_{i}$. We assume that each player's strategy set $S_{i}$ is a complete pseudo lattice and that $S$ is endowed with the product order. A strategy profile $s=(s_{i})_{i\in I}$ is a (pure-strategy) {\it Nash equilibrium} if $u_{i}(s)\ge u_{i}(s_{i}^{\prime},s_{-i})$ for every $i\in I$ and $s_{i}^{\prime}\in S_{i}$. 


We say that $\Gamma$ is a \textbf{pseudo quasi-supermodular game} if, for all $i\in I$,
\begin{enumerate}
    \item[(P1)] $u_{i}$ is bounded, and  order upper semicontinuous in $s_i$ for each $s_{-i}\in S_{-i}$;
    \item[(P2)] $u_i$ is pseudo quasi-supermodular in $s_i$ and satisfies  single-crossing in $(s_i,s_{-i})$.
\end{enumerate}

This class generalizes the quasi-supermodular games of \cite{milgrom/shannon:94}. It relaxes the crucial requirements that strategy spaces be lattices and that best-response correspondences form sublattices. Drawing on our individual choice results (\Cref{sec:ind-choice}), conditions (P1) and (P2) together ensure that each player's best-response correspondence is pSS monotonic. This allows us to apply our fixed-point results (\Cref{sec:fixed point})  to establish the existence and comparative statics of equilibria.

As a cardinal specialization of a pseudo quasi-supermodular game, we say that $\Gamma$ is a \textbf{pseudo supermodular game} if, for all $i\in I$,  (P1) holds and 
\begin{enumerate}[label=(\roman*')]
   \item[(P2$'$)] $u_i$ is pseudo supermodular in $s_i$ and satisfies increasing-differences in $(s_i,s_{-i})$.
\end{enumerate} Pseudo supermodular games prove useful in later sections when we study mixed-strategy Nash equilibria and perfect equilibria. This is because, as a cardinal concept, the property (P2$'$) is preserved under randomization.

We next establish the existence, structure, and monotone comparative statics of Nash equilibria for  pseudo quasi-supermodular games. For this purpose, we parametrize the players' payoff functions as $u_i(\cdot, t): S \to \mathbb{R}$ and $u(\cdot, t) = (u_i(\cdot, t))_{i \in I}$. Letting $\Gamma(t) = (I, S, u(\cdot, t))$, we denote the set of pure Nash equilibria of $\Gamma(t)$ by $\Eq(t)$.

\begin{proposition} \label{prop:NE_existence_MCS}  For a family of pseudo quasi-supermodular   games  $\Gamma (t)$,
\begin{description} 
\item[(i)]  the set of (pure)  Nash equilibria $\Eq(t)$ is a nonempty complete pseudo lattice; 
\item[(ii)]  if $u_{i}$ satisfies single-crossing in $(s_i,t)$ for all $i \in I$, then $\Eq(t')\ge_{ws}\Eq(t)$ for all $t'>t$.
\end{description}
\end{proposition}

The proof builds directly on our earlier results. By \Cref{thm:pseudo-argmax}, each player's best-response correspondence $B_i(s_{-i}, t)$ is a nonempty complete pseudo sublattice. By \Cref{cor:MCS-sublat}, it is also pSS monotonic in $s_{-i}$. The joint best-response correspondence $B(\cdot, t)$ is therefore pseudo monotonic, so applying \Cref{thm:tarski-correspondence} establishes that the equilibrium set $\Eq(t)$ is a nonempty complete pseudo lattice. The comparative statics result in part (ii) follows from \Cref{thm:mcs of fp}.



\paragraph{Application to generalized Bertrand games.}

We now apply our equilibrium framework to a class of generalized Bertrand games. Consider a finite set of firms $I$, where each firm $i \in I$ chooses a price $p_i$ from a finite set $P_i \subset \mathbb{R}_+$. Given a price profile $p \in P := \prod_{j \in I} P_j$, firm $i$ sells $D_i(p_i, p_{-i})$ units of its product at a cost determined by an increasing, convex function $C_i: \mathbb{R}_+ \to \mathbb{R}_+$. Firm $i$'s payoff  is its profit
\begin{equation} \label{eq:bertrand_profit}
U_{i}(p):=p_{i}D_{i}(p)-C_{i}(D_{i}(p)).
\end{equation}
We assume that each firm's demand function $D_i: P \to \mathbb{R}_+$ satisfies:
\begin{itemize}
    \item [(D1)]  $D_{i}$ is weakly decreasing in $p_{i}$ and weakly increasing in $p_{-i}$; 
    \item [(D2)]$\frac{D_{i}(p_{i}^{\prime},p_{-i})}{D_{i}(p_{i},p_{-i})}\le\frac{D_{i}(p_{i}^{\prime},p_{-i}^{\prime})}{D_{i}(p_{i},p_{-i}^{\prime})}$ for any $p_{i}<p_{i}^{\prime}$ and $p_{-i}<p_{-i}^{\prime}$ such that $D_{i}(p_{i},p_{-i})>0$.
\end{itemize}
 
The second part of (D1) indicates that the firms' products are substitutes. Condition (D2) strengthens this property by implying that firm $i$'s demand becomes more inelastic as rivals' prices increase.

This framework generalizes the Bertrand games studied by \citet{milgrom/shannon:94}, who assume a stronger version of (D2) requiring $D_i$ to be strictly positive and continuously differentiable.\footnote{Strictly speaking, \citet{milgrom/shannon:94} do not assume the finiteness of $P_{i}$. It ensures condition (P1) required of (lower) pseudo quasi-supermodular games and ensures the existence of a Nash equilibrium in pure strategies.} In contrast, our condition applies only when demand is strictly positive. Notably, our weaker requirement encompasses pure Bertrand games, which their analysis excludes.\footnote{More formally, a generalized Bertrand game is a pure Bertrand game if, for each i, $C_{i}(q)=c_{i}q$ for some $c_{i}\in[0,\max_{p_{i}\in P_{i}}p_{i}]$, and $D_{i}(p)=1/|\arg\min_{j\in I}p_{j}|$ if $p_{i}=\min_{j\in I}p_{j}$ and $D_{i}(p)=0$  otherwise. Then, one can show that a pure Bertrand game is a generalized Bertrand game; see   \Cref{lem:pure-bertrand} in \Cref{app:generalized-bertrand} of the Supplementary Appendix.} 
Beyond pure Bertrand competition, this class includes various games where demand may drop to zero at certain price profiles or where demand functions are discontinuous—scenarios typically excluded from the existing monotone comparative statics literature.

Notably, a generalized Bertrand game is not necessarily pseudo quasi-supermodular, so our earlier theorems do not apply directly. To see this, consider a standard Bertrand game with two firms and focus on firm 1’s payoff when it has constant marginal cost $c_1$. If firm 2 sets $p_2 < c_1$, then any $p_1 > p_2$ yields zero demand and thus maximizes firm 1’s profit at zero. By contrast, if firm 2 raises its price to some $p_2' > c_1$, then firm 1 can earn strictly positive profit, and no price $p_1 > p_2'$ maximizes its profit (assuming $P_1$ is sufficiently dense). Consequently, firm 1’s payoff fails the single-crossing property, meaning that the game is not pseudo quasi-supermodular.


Despite this failure, the game satisfies the ``lower'' requirements of a pseudo quasi-supermodular game---the violation occurs only in the ``upper'' single-crossing condition. Thus, our directional results (\Cref{prop:NE_existence_MCS-lower} in the Appendix) apply, ensuring both the existence and comparative statics of equilibria.\footnote{To intuitively see why this directional  generalization holds, note that our equilibrium existence and comparative statics results are derived from the properties of best-response correspondences and the associated fixed-point theorems. As established in \Cref{sec:ind-choice} and \Cref{sec:fixed point}, both the individual choice results and the fixed-point theory admit directional counterparts that remain valid in this setting. A formal statement and proof are provided in the Appendix.}

For our comparative statics analysis, consider a family of games $\Gamma(t) = (I, P, (U_i(\cdot, t))_{i \in I})$ satisfying:
\begin{itemize}
\item[(B1)] If $D_i(p_i, p_{-i}, t) > 0$, then $D_i(p_i, p_{-i}, t') > 0$ and $\frac{D_i(p_i', p_{-i}, t')}{D_i(p_i, p_{-i}, t')} \ge \frac{D_i(p_i', p_{-i}, t)}{D_i(p_i, p_{-i}, t)}$ for all $t' > t$ and $p_i' > p_i$.
\item[(B2)] $C_i(q', t) - C_i(q, t)$ is weakly increasing in $t$ for all $q' > q$.
\end{itemize}   
Intuitively, a higher parameter $t$ corresponds to more inelastic demand or higher marginal costs. To derive payoff implications, we consider a slight strengthening of condition (B2):
\begin{enumerate}
\item[(B2$'$)] For $t < t'$, $C_i(q, t) = c_i q \le c_i' q = C_i(q, t')$ with $\max_{p_i \in P_i} p_i \ge c_i$, and $D_i(p, t) \le D_i(p, t')$ for all $q \in \mathbb{R}_+$ and $p \in P$.
\end{enumerate}
Condition (B2$'$) implies condition (B2). With these preparations, we state our results for generalized Bertrand games below.

 
\begin{corollary}  \label{cor:general_bertrand} For a  family  of generalized Bertrand games  $\Gamma (t)$, 
\begin{description}
    \item[(i)]  if   (B1) and (B2) hold for each $i \in I$, then $\Eq(t)  \ne \emptyset$ and $\Eq(t')\ge_{lws}\Eq (t)$ for all  $t' > t$; 
    \item[(ii)]   if   (B1) and  (B2$'$) hold for each $i \in I$,  then the set of equilibrium profits for each firm $i$ with $c_{i}'=c_{i}$ in $\Gamma (t')$ lower weak-set dominates that in $\Gamma (t)$ for all $t' > t$.
\end{description}      
\end{corollary}

\subsection{Mixed-strategy Nash equilibria} \label{sec:mixed} 

A long-standing challenge in the theory of monotone comparative statics is its extension to mixed strategies.  While mixed strategies are essential for game-theoretic analysis,  the set of mixed strategies $\Delta(S_i)$, ordered by first-order stochastic dominance, generally fails to form a lattice---even when the underlying pure-strategy space $S_i$ is a lattice; recall \Cref{KKO example}.\footnote{In this example, one can take $X = \Delta(\{0,1\}^2)$ to represent player $i$’s mixed-strategy space over the pure strategy set $\{0,1 \}^2$.}  \cite{echenique2003mixed} explicitly identifies this non-lattice structure in multidimensional settings as the primary reason why the standard complementarity framework has struggled to accommodate mixed strategies.

By dispensing with the lattice requirement and utilizing our fixed-point results for pseudo lattices, we provide a unified framework for the MCS analysis of mixed strategies. This section also serves as a crucial intermediate step for our analysis of perfect equilibria in \Cref{sec:wmcs of thpe}. Establishing tools to handle mixed strategies enables us to conduct MCS on refinements, such as perfect equilibria, which require perturbing to fully mixed strategies.

We restrict our attention to a pseudo supermodular game $\Gamma=(I,S,u)$, as cardinal payoff properties such as pseudo supermodularity are preserved under randomization. We assume each pure-strategy set $S_i$ is a compact, partially ordered Polish space containing the largest and smallest elements.     Let $\mG=(I, \Sigma, \overline{u})$ be the mixed extension of $\Gamma$, with $\Sigma:= \times_{i \in I} \Sigma_i$ and $\bar u:= (\bar u_i)_{i \in I}$, where $\Sigma_i = \Delta(S_i)$ is the set of player $i$'s mixed strategies and $\overline{u}_i: \Sigma \to \mathbb{R}$ is the expected payoff defined by:
\begin{equation} \label{expected-utility-mixed}
\overline{u}_{i}(\sigma):=\int u_{i}(s)\sigma(ds), \sigma\in\Sigma.
\end{equation}
Each $\Sigma_i$ is partially ordered by the first-order stochastic dominance relation $\ge^{sd}$. Further,  since each $\Sigma_i$ is compact and contains extremal elements, it forms a complete pseudo lattice by \Cref{cor:compact-with-ext}.  

A mixed strategy profile $\sigma=(\sigma_{i})_{i\in I}$ is a mixed-strategy Nash equilibrium if $\bar u_{i}(\sigma)\ge \bar u_{i}(\sigma_{i}^{\prime},\sigma_{-i})$ for every $i\in I$ and $\sigma_{i}^{\prime}\in \Sigma_{i}$. 
For a parametrized game $\G(t) = (I, S, u(\cdot, t))$  and its mixed extension $\mG(t)=(I,\Sigma,\bar u(\cdot, t))$, we let   $\NE(t)$ denote the set of all Nash equilibria of $\mG (t)$---that is, the set of all pure and mixed Nash equilibria of $\G (t)$.

\begin{theorem} \label{thm:mcs-mixed}
For a family of  pseudo supermodular games $\Gamma(t)$ and their mixed extensions $\mG(t)$,
\begin{description}
    \item[(i)] $\NE(t)$ possesses the largest and smallest elements,  which are both pure;
    \item[(ii)] if each $u_i$ satisfies increasing differences in $(s_i,t)$, then $\NE(t') \ge_{ws} \NE(t)$ for all $t' > t$.
\end{description}
\end{theorem}

The proof leverages the complete pseudo lattice structure of the mixed-strategy space. A primary difficulty, as discussed in \Cref{sec:decision under uncertainty} and \Cref{ex:expected-utility-fail}, is that the player's best-response correspondence 
\begin{equation} \label{br_mixed_nash}
B_{i}(\sigma_{-i}) := \arg \max_{\sigma_i \in \Sigma_i} \overline{u}_i(\sigma_i, \sigma_{-i})
\end{equation}
need not be pSS monotonic in rivals' mixed strategies. However, by \Cref{cor:random_choice}, each $B_i(\sigma_{-i})$ possesses largest and smallest elements that are pure strategies. These extremal best responses are nondecreasing in $\sigma_{-i}$ and therefore pSS monotonic. We can thus ``sandwich'' the full best-response correspondence between these monotonic extremal selections. Applying \Cref{thm:tarski-correspondence} and \Cref{thm:mcs of fp} then establishes the existence and comparative statics of the largest and smallest equilibria, both of which are pure.\footnote{This ``sandwiching'' of mixed equilibria between extremal pure equilibria was identified by \citet[Theorem 5]{milgrom/roberts:90} for lattice-based supermodular games. However, their approach relies on a serial undomination argument that requires full continuity of payoffs. In contrast, our approach utilizes monotone shifts in best-response correspondences on a pseudo lattice, establishing existence and MCS under weaker continuity assumptions and on more general domains.}

\section{Perfect Equilibria}
\label{sec:wmcs of thpe}

The existing monotone comparative statics analyses of games have been largely confined to Nash equilibria.\footnote{See, for instance, \citet{echenique2004extensive}, which develops a notion of supermodularity for dynamic games and applies it to subgame-perfect equilibrium.} 
 Extending the analysis to refinements such as (trembling-hand) perfect equilibria is important. For example, classical games motivating perfection are supermodular.\footnote{For instance, the two games below are supermodular.
\[
\begin{array}{c|cc}
\multicolumn{3}{c}{\text{Coordination Trap}} \\
 & L & R \\ \hline
T & 0,0 & 1,1 \\
B & 0,0 & 0,0
\end{array}
\qquad
\begin{array}{c|cc}
\multicolumn{3}{c}{\text{Entry Deterrence}} \\
 & Fight & Accom. \\ \hline
Enter & -1,-1 & 1,1 \\
Not & 0,2 & 0,2
\end{array}
\]
The profiles $(B,L)$ and $(\mbox{Not, Fight})$ of the respective games are Nash but not perfect.}

However, such an extension presents a fundamental challenge: it requires dealing with perturbations in fully mixed strategies. As discussed earlier, however, the space of mixed strategies does not form a lattice---let alone a complete lattice---making the standard lattice-based machinery inapplicable. Our approach, which  accommodates a more general domain structure, overcomes this limitation and establishes both the existence of perfect equilibria in pure strategies and their monotone comparative statics properties.





Let us adopt  the same framework as  \Cref{sec:mixed}  and   consider a  pseudo supermodular  game $\G = (I, S, u)$.  Let    $\mG = (I, \Sigma, \bar{u})$ denote the mixed extension of $\Gamma$, where each $\Sigma_i = \Delta(S_i)$ is partially ordered by the first-order stochastic dominance relation $\ge^{sd}$ and $\bar u = (\bar u_i)_{i \in I}$ with $\bar u_i$ defined in \eqref{expected-utility-mixed}.  

Following \citet{selten:75}, we consider \emph{perturbed} games in which players are constrained to play \emph{fully mixed} strategies. 
To formalize this idea in general (not necessarily finite) games, for each player $i$, let $\mathcal S_i$ be the Borel $\sigma$-algebra on $S_i$, and let $\mathcal M_i$ denote the set of all nonnegative measures $\mu_i$ on $\mathcal S_i$ satisfying $\mu_i(S_i)\le 1$. For two measures $\mu_i, \mu_i' \in \mathcal{M}_i$, let us write  $\mu_i' \sqsupseteq \mu_i$ if  $\mu_i'(S_i') \ge \mu_i(S_i')$ for all $S_i' \in \mathcal S_i$.  
Let $\mathcal M_i^0 \subset \mathcal M_i$ denote the set of full-support measures, i.e.,  the set of all measures $\mu_i \in \mathcal M_i$ such that $\mu_i(S_i') > 0$ for every nonempty open set $S_i' \in \mathcal S_i$.  
Let $\mathcal M := \times_i \mathcal M_i$ and $\mathcal M^0 := \times_i \mathcal M_i^0$.

For $\mu = (\mu_i)_{i\in I}\in \mathcal M$, define
$$
\Sigma_i^{\mu} := \{\sigma_i\in \Sigma_i : \sigma_i \sqsupseteq \mu_i\},
$$
the set of player $i$’s mixed strategies that place at least the measure $\mu_i$ on every Borel subset of $S_i$.  
As shown in the Online Appendix, $\Sigma_i^{\mu}$ inherits the complete pseudo lattice structure of $S_i$.\footnote{See \Cref{lem:constrained game st space} in the Online Appendix for the proof.}  
When $\mu \in \mathcal M^0$, the strategy space $\Sigma^\mu$ captures the idea that players are constrained to play fully mixed strategies---that is, each player $i$ must assign positive probability to every open subset of $S_i$. Let $\mG^{\mu} := (I,\Sigma^{\mu},\bar u)$ denote the $\mu$-constrained game of $\mG$, where $\Sigma^{\mu} = \times_i \Sigma_i^{\mu}$.

To quantify the size of  perturbations, define for each $\mu_i\in\mathcal M_i$ the total variation norm
$$\|\mu_i \|=\sup_{S_i'\in\mathcal S_i}|\mu_i(S_i')|=\mu_i(S_i),$$
and for each profile $\mu=(\mu_i)_i$, set $\| \mu \|=\max_i \|\mu_i \|$.

We are now ready to define perfect equilibrium: A strategy profile $\sigma$ of the game $\mG = (I, \Sigma, \bar u)$ is a \textbf{perfect equilibrium} if there exists a sequence of $\mu^n$-constrained games $\mG^{\mu^n}$ with $\mu^n \in \mathcal M^0$ and $\|\mu^n \| \to 0$ such that each $\mG^{\mu^n}$ admits a Nash equilibrium $\sigma^n$ converging weakly to $\sigma$.

Our formulation of constrained games and perfect equilibrium parallels Selten’s and coincides with his original definition when strategy spaces are finite. For games with infinite strategy spaces, our definition corresponds to \citet{simon1995equilibrium}’s notion of strong perfect equilibrium.\footnote{These authors advocate this notion because it preserves the “hallmark” property of \emph{limit admissibility}, meaning that a strategy places no mass in the interior of the set of weakly dominated strategies. In finite games, limit admissibility coincides with admissibility (that requires weakly dominated strategies to be played with zero probability). \citet{simon1995equilibrium} also define a weaker concept, called weak perfect equilibrium, which fails to satisfy limit admissibility.}

In the setting of infinite normal-form games, it refines the notion of Nash equilibrium when payoff functions are continuous:\footnote{To see that payoff continuity cannot be dispensed with, consider a game in which each player $i=1,2$ chooses $x_i\in [0,1]$ simultaneously, and the payoffs are $(1,1)$ for all $(x_1,x_2)$, except when $x_1=x_2=1$, in which case both players get zero payoffs. Every pair $(x_1,x_2)\ne (1,1)$ is a perfect equilibrium, so a limit point $(x_1,x_2)=(1,1)$ is perfect as well.  But it is not a Nash equilibrium.}
\begin{lemma}\label{lem:perfect is nash}
If each $u_i$ is continuous, then any perfect equilibrium is also a Nash equilibrium.
\end{lemma}
\noindent The payoff continuity is used only for \Cref{lem:perfect is nash} in this section; the subsequent results do not require this assumption.

To analyze the  constrained game $\mG^{\mu}$, it is useful to study a particular class of strategies in this game. We say a (mixed) strategy $\sigma_i \in \Sigma_i^{\mu}$ is \textbf{constrained-pure at} $s_i \in S_i$  if  $\sigma_i$ puts maximal feasible mass on $s_i$:  
 for each $S_i' \in \mathcal{S}_i$,  \begin{align}
\sigma_i (S_i')  =\begin{cases}
1-\mu_i (S_i \setminus S_i')  &\mbox{ if } s_i \in S_i' \\ \mu_i (S_i')  & \mbox{ otherwise.}
\end{cases} \label{constrained-pure} \end{align}  It is straightforward to see that this strategy belongs to  $ \Sigma_i^\mu $.\footnote{To see that $\sigma_i$  is a probability measure, observe that for any $S_i' \subset S_i$, $$\sigma_i(S_i') = (1- \mu_i (S_i\setminus \{s_i\})) \delta_{s_i} (S_i') + \mu_i (S_i' \setminus \{s_i\}),$$ where $\delta_{s_i}$ denotes the Dirac measure at $s_i$.
This expression represents a nonnegative linear combination of the two nonnegative measures $\delta_{s_i} (\cdot)$ and $\mu_i (\cdot \setminus \{s_i\})$, and hence defines a nonnegative measure. Moreover, since $\sigma_i (S_i) = (1- \mu_i (S_i\setminus \{s_i\}))  + \mu_i (S_i \setminus \{s_i\}) =1$, it is a probability measure.}  We call any strategy that is constrained-pure at some pure strategy a {\bf constrained-pure} strategy.

The following result is central to our comparative statics analysis of perfect equilibria. It establishes the existence of the largest and smallest Nash equilibria of the constrained games and their comparative statics.  To do so, we consider a family of pseudo supermodular games  $\G (t)  = (I,  S, u (\cdot, t) )$ and their $\mu$-constrained games $\mG^\mu (t)  = (I, \Sigma^\mu, \bar u (\cdot, t))$. The set of all Nash equilibria of $\mG^\mu (t)$ is denoted by $\NE^\mu (t)$. 
\begin{proposition}
\label{lem:perfect_eq_exist_comp}  For a family of pseudo supermodular games  $\G (t) $ and for each $\mu \in \mathcal{M}$,  
\begin{description}
 \item[(i)]  $ \NE^\mu (t)$ contains  the largest and smallest elements,  which are constrained-pure;
 \item[(ii)]  if  each $u_i$ satisfies increasing-differences in $(s_i,t)$, then $\NE^\mu(t')\ge_{ws}\NE^\mu(t)$ for all $t' > t$. 
\end{description}
\end{proposition}

The logic for this result parallels that for mixed-strategy Nash equilibria  (\Cref{thm:mcs-mixed}), except that here we work with the constrained mixed-strategy space  $\Sigma^{\mu}$.
In this setting, the extremal best responses take the form of constrained-pure strategies.
Apart from considering these constrained-pure best responses, the reasoning proceeds analogously: we ``sandwich’’ the full best-response correspondence between the extremal constrained-pure best responses and apply our  fixed-point results (\Cref{cor:tarski} and \Cref{thm:mcs of fp})  to establish the existence and monotone comparative statics of the extremal equilibria.

We now establish our main results of the current section, namely,  the existence and comparative statics of perfect equilibria: 

\begin{theorem}\label{thm: main result for thpe} For a family of pseudo supermodular games $\G (t)$,
\begin{description}
        \item[(i)] $\G (t)$ has a perfect equilibrium in pure strategies; 
    \item[(ii)] the set of  perfect equilibria of $\G (t)$ is compact and contains maximal/minimal elements, which are all pure; 
    \item[(iii)] if each $u_i$  satisfies increasing differences in $(s_i, t)$,   perfect equilibria of $\G (t')$ weak-set dominate those of $\G (t)$ for all $t'> t$.      
\end{description}
\end{theorem}

We briefly sketch the proof; the full argument appears in the Appendix.
Existence follows directly from \Cref{lem:perfect_eq_exist_comp}(i), which guarantees that each constrained game $\mG^{\mu} (t)$ admits a constrained-pure Nash equilibrium.
For any sequence of full-support measures ${\mu^n}$ converging to zero, the corresponding constrained-pure equilibria  converge (along a subsequence) to a pure strategy profile that forms a perfect equilibrium.

Part (ii) also builds on \Cref{lem:perfect_eq_exist_comp}(i).
Any perfect equilibrium can be obtained as the limit of Nash equilibria of constrained games for some vanishing sequence ${\mu^n}$.
Each such equilibrium is sandwiched between the smallest and largest constrained-pure equilibria of $\mG^{\mu^n} (t)$, and their limits yield pure perfect equilibria that bound the given equilibrium.
We do not, however, claim the existence of globally smallest or largest perfect equilibria, since different sequences ${\mu^n}$ may generate incomparable limits.

The comparative statics result in part (iii) follows from \Cref{lem:perfect_eq_exist_comp}(ii).
The constrained-pure equilibria of perturbed game $\mG^{\mu^n} (t')$ weak-set dominate those of $\mG^{\mu^n} (t) $ for each $n$, and this comparison is preserved after taking limits.

The existence of perfect equilibria in general infinite games is also established by  \citet{simon1995equilibrium} under the full continuity of payoffs.
In contrast, our result ensures the existence of \emph{pure} perfect equilibria and, more importantly, delivers a monotone comparative statics  of perfect equilibria  absent in \citet{simon1995equilibrium}.

\section{Concluding Remarks}

This paper has revisited the order-theoretic foundations of monotone comparative statics. The existing theory has long relied on the assumption that the domain of choice forms a lattice. We have shown that this structural requirement is not essential in many settings. By introducing the weaker notion of a \textit{pseudo lattice}---which requires only the existence of minimal upper bounds and maximal lower bounds---we have generalized the core machinery of the theory, including the Monotonicity Theorem for individual choice and Tarski's fixed-point theorem.

The primary contribution of this generalization is its capacity to handle environments that fail the lattice property, most notably the space of probability distributions. This flexibility has allowed us to provide a unified framework for analyzing mixed-strategy Nash equilibria and, significantly, to conduct the first general monotone comparative statics analysis of (trembling-hand) perfect equilibria. By treating perfect equilibria as limits of Nash equilibria in constrained games---where strategy spaces are pseudo lattices but not lattices---we established the existence of pure perfect equilibria and their monotonicity with respect to the underlying environment.

Our framework opens several avenues for future research. First, while we have established the existence of maximal and minimal perfect equilibria, our current results do not guarantee the existence of a unique \textit{largest} or \textit{smallest} perfect equilibrium. Determining the conditions under which the set of perfect equilibria admits these global extremal elements remains an open question.

Second, the applicability of our framework to probability measures suggests natural extensions to \textit{Bayesian games}. Since the space of distributional strategies often lacks a lattice structure under standard orders, our pseudo-lattice approach could facilitate monotone comparative statics analysis in games of incomplete information.

Finally, our results may prove useful in the field of \textit{information design}. As illustrated in our examples, the set of information structures ordered by the mean-preserving spread (convex order) forms a complete pseudo lattice but generally fails to be a lattice. Applying our  optimization and fixed-point theorems to this domain could yield new insights into how optimal information structures respond to changes in the economic environment.

\newpage

\bibliographystyle{economet}
\bibliography{bibmatching}

\appendix 
 
\section{Proofs for \Cref{sec:prel}}

\begin{proof} [Proof of \Cref{thm:pseudo-lattice-chain-complete}] [{\bf ``$\Leftarrow$'' direction:}] For any nonempty set $S\subset X$, we will show that  $\Meet_X S \neq \emptyset$ (To prove $\Join_X S  \neq \emptyset$ is analogous and thus omitted). Let $\underline x$ be the smallest point of $X$, which exists by assumption. For $S$, their common lower bound, denoted by  $L_S$, is nonempty because $\underline x \in L_S.$  For any chain   $C\subset L_S$, there exists a supremum $z$ of  $C$ in $X$ since $X$ is chain complete, by assumption.  Then, $x \ge c$ for any $x\in S$ and  $c \in C$ since $C \subset L_S$. This and the fact that $z$ is the supremum of $C$ imply that $x \ge z$ for each $x\in S$.  Therefore, $z \in L_S$. Thus, by Zorn's lemma, the set $L_S$ has a maximal element, that is, $\Meet_X S\neq \emptyset,$ as desired.

 [{\bf ``$\Rightarrow$'' direction:}] Assume $X$ is a complete pseudo lattice.  Then, $\Join_X X$ is nonempty. Let $x \in \Join_X X$. Then,  since $\join_X X \subset U_X$, we have $x \ge x'$ for every $x' \in X$, which shows that $x$ is the largest element. To prove the  existence of the smallest element is 
analogous and thus omitted. \end{proof}

\begin{proof} [Proof of \Cref{lem:compact-chaincomp} ] We prove that every nonempty chain $C\subset X$ admits a supremum (and analogously an infimum).  To this end, fix any  nonempty chain $C\subset X$ and let $U:=\bigcap_{x\in C}U_x$ be the set of all upper bounds for $C$, where we recall $U_x=\{y\in X: y\ge x\}$.

\begin{claim} \label{claim:comp-chain1} $U\ne \emptyset$.  
\end{claim}
\begin{proof}  Suppose not.  Then, $\bigcap_{x\in C}U_x=\emptyset$.  Hence, $\bigcup_{x\in C} (X\setminus U_x)=X$. Recall, by the definition of natural topology,  $X\setminus U_x$ is open.  Since $X$ is compact, there exists a finite set $C_f\subset C$ such that 
\begin{equation} \label{eq:comp-chain1}
      \bigcup_{x\in C_f} (X\setminus U_x)=X. 
\end{equation}

Meanwhile, since $C_f$ is a finite chain, it admits a maximum $\bar x$. Since $X\setminus U_{\bar x}\supset X\setminus U_{ x}$ for all $x\in C_f$, we have 
\begin{equation} \label{eq:comp-chain2}
      \bigcup_{x\in C_f} (X\setminus U_x)=X\setminus U_{\bar x}. 
\end{equation}
It follows from \eqref{eq:comp-chain1} and \eqref{eq:comp-chain2} that $U_{\bar x}=\emptyset$, which, however, contradicts $\bar x\in U_{\bar x}$.
\end{proof}

For each $c\in C$ and $u\in U$, let $[c,u]:=\{x\in X:  c\le x\le u\}.$   \Cref{claim:comp-chain1} ensures that this set is well defined.

\begin{claim} \label{claim:comp-chain2}$\bigcap_{c\in C, u\in U }[c, u]\ne \emptyset$.
\end{claim}

\begin{proof} Suppose not.  Then, $X\setminus(\bigcap_{c\in C, u\in U }[c, u])=X$, so
$$\bigcup_{c\in C, u\in U }(X\setminus [c, u])=X.$$
Given our topology, $[c,u]$ is closed, so $X\setminus [c, u]$ is open for each $c\in C, u\in U$.  By the compactness of $X$, we have $(c^1,u^1), ..., (c^K,u^K)$ in $C\times U$, for some $K\in \mathbb{N}$, such that 
$$\bigcup_{k=1 }^K(X\setminus [c^k, u^k])=X,$$ 
so 
    $$\bigcap_{k=1 }^K [c^k, u^k]=\emptyset.$$

Meanwhile, since $C$ is a chain, there exists $\bar c:=\max\{c^1, ..., c^K\}$, so 
\begin{equation}
    \bigcap_{k=1 }^K [c^k, u^k]= \bigcap_{k=1 }^K [\bar c, u^k] =\emptyset.
\end{equation}
However, since $U$ consists of upper bounds of all $c\in C$, $\bar c\le u^k$ for all $k=1, ..., K$, a contradiction.
\end{proof}

To complete the proof, let $a\in \bigcap_{\tiny c\in C, u\in U } [c, u]$, which is possible by \Cref{claim:comp-chain2}.  Since $c\le a$ for all $c\in C$ and since $a\le u$ for all $u\in U$, we conclude that $a=\sup C$.  \end{proof}

\section{Proof for \Cref{sec:ind-choice}}\label{app:ind-choice}

\begin{proof} [Proof of \Cref{thm:MCS-UPSS}]   In the following, we prove the ``upper'' version of the result: the ``lower'' version is analogous and thus omitted.

\noindent [{\bf ``$\Leftarrow$'' direction:}] 
To show $M_{S'} (t') \ge_{uwpss} M_{S} (t)$ for every $t,t' \in T$ with $t' \ge t$ and $S, S' \subseteq X$ with  $S'\ge_{wpss} S$, consider any $x \in M_{S} (t)$ and  $x' \in M_{S'} (t')$.  Consider any   $z '\in (x \join x')\cap S'$ and $z\in (x \meet x')\cap S$; such $z$ and $z'$ exist  since $S'\ge_{wpss} S$.
We have
\begin{align*}
    u(x,t)\ge u(z,t) \Rightarrow u(z',t)\ge u(x',t)  \Rightarrow u(z',t') \ge u(x',t'),  
\end{align*}
where the first implication follows from upper pseudo quasi-supermodularity of $u$ and the second implication follows because $u$ satisfies  upper single-crossing and $z' \ge x'$. Because $x' \in M_{S'} (t')$ by assumption, it follows that $z'\in  M_{S'} (t')$, as desired.


\noindent  [{\bf ``$\Rightarrow$'' direction:}]  To prove that $u$ is upper single-crossing in $(x,t)$, consider any  $t' \ge t$ and $x' \ge  x$ such that $u(x',t)\ge u(x,t)$.  Choose $S=S'=\{x, x'\}$.  Clearly, $S'\ge_{wpss} S$, and $x'\in M_S(t)$. Fix any $x''\in M_{S'}(t') \subset \{x, x'\}$. That $M_{S'} (t') \ge_{uwpss} M_S (t)$ and $x' \join x'' =\{x'\}$ implies $x'\in M_{S'} (t')$, which in turn implies $u(x',t')\ge u (x,t')$, as desired.

Next, to prove $u$ is upper pseudo quasi-supermodular in $x$, let $S=\{x, z\}$ and $S'=\{ x', z'\}$, for arbitrary $z\in x\meet x'$ and $z'\in x\join x'$.  Suppose that $x$ and $x'$ are incomparable (since otherwise the result holds trivially). Assume that $u(x,t)\ge u(z,t)$.  Then, $x\in M_S(t)$. Since $S' \ge_{wpss} S$ and $S' \cap (x \join x') = \{z'\}$,   $M_{S'} (t) \ge_{uwpss} M_S (t)$ requires   $z' \in M_{S'}(t)$. It follows that $u(z',t)\ge u(x',t)$, proving the upper pseudo quasi-supermodularity of $u$ in $x$. \end{proof}

\begin{proof}[Proof of \Cref{thm:pseudo-argmax}]
This result follows immediately from  a directional version  in \Cref{thm:pseudo-argmax-upper} below. \end{proof} 

We say that $S\subset X$ is an  \emph{upper (resp. lower) pseudo sublattice} if, for any $x,x' \in S$, $x \join x' \subset S$ (resp. $x\meet x' \subset S$), and that $S$ is a \emph{complete  upper (resp. lower) pseudo sublattice} if, for every nonempty $S' \subseteq S$, $\Join_X S'$  (resp. $\Meet_X S'$) is a nonempty subset of $S$.

\begin{manualtheorem}
{4$'$} \label{thm:pseudo-argmax-upper} Assume $X$ is a pseudo lattice and $u: X \to \mathbb R$ is   upper (resp. lower) pseudo quasi-supermodular.
\begin{description}
    \item[(i)] $\arg\max_{x\in X} u (x)$ is  an upper (resp. lower) pseudo sublattice of $X$ whenever it is nonempty.
    \item [(ii)] In addition, if $X$ is a complete pseudo lattice and $u$ is order upper semicontinuous, then $\arg\max_{x\in X} u (x)$ is a nonempty, complete upper (resp. lower) pseudo sublattice, admitting the largest (resp. smallest) point.  
\end{description}  
\end{manualtheorem}
\begin{proof}  Throughout the proof, we only establish the upper case since the proof for the lower case is analogous.

 For (i), if $s,s' \in \arg\max_{x\in X} u (x)$, then  $u (s) \ge u (z)$ for any $z \in s\meet s'$, which implies by upper pseudo quasi-supermodularity that   $u (z')\ge  u (s')$ for any $z' \in s\join s'$, meaning $z' \in \arg\max_{x\in X} u (x)$, as desired. 

For (ii), let us establish  a couple of claims (see the Online Appendix for proof): 
\begin{claim} \label{claim:chain}
    For any subset $X' \subset X$ and each $x \in X'$, there is a maximal chain in $X'$  containing $x$.  
\end{claim}

\begin{claim} \label{claim:nonempty-join}
If $X$ is chain complete, then for any $x,y$ and $z$ with  $x,y \in U_z$,   $(x \meet y) \cap U_z $ is nonempty. Also, for any $x,y$ and $z$ with $x, y \in L_z$,  $(x \join y) \cap L_z$ is nonempty.
\end{claim}

To first prove  $\arg\max_{x\in X} u (x) \ne \emptyset$, let
\[
 M := \sup_{x \in X} u(x),
\]
(where $M$ is possibly infinite) and observe that there exists a sequence $\{x_n\}$ such that
\[
\lim_{k \to \infty} \sup_{n \ge k} u(x_n) = M .
\]
We construct an increasing sequence (hence a chain) $\{y_n\}$ satisfying $u(y_n) \ge u(x_n)$ for every $n$. Letting $y$ denote the supremum of $\{y_n\}$, the order upper semicontinuity of $u$ implies
\[
u(y)
\ge \lim_{k \to \infty} \sup_{n \ge k} u(y_n)
\ge \lim_{k \to \infty} \sup_{n \ge k} u(x_n)
=  M,
\]
which shows that $y \in \arg\max_{x \in X} u(x)$.

To construct $\{y_n\}$, we begin by defining  $z^0 = \{ z^0_n\}_{n \in \mathbb{N}}$  such that $z^0_n = x_n$. Given $z^{m-1} = \{z_{n}^{m-1}\}_{n \in \mathbb{N}}$,  we recursively construct, for each $m \ge 1$, a sequence $z^m = \{z_n^m\}_{n \in \mathbb{N}}$ satisfying:
\begin{itemize}
    \item[(a)] $z_n^m = z_n^{m-1}$ for all $n < m$;
    \item[(b)] $z_1^m \le z_2^m \le \cdots \le z_m^m \le z_n^m$ for all $n > m$;
    \item[(c)] $u(z_n^m) \ge u(x_n)$ for all $n \ge 1$.
\end{itemize}
Note that $z^0$ satisfies (a)--(c)  trivially.  

Once such sequences are constructed, it suffices to define $y_n := z_n^n$. Indeed, for every $n$, conditions  (a) and (b) imply that
$z_{n+1}^{n+1} \ge z_n^n$, so $\{y_n\}$ is increasing, and condition (c) implies that
$u(z_n^n) \ge u(x_n)$.

Fix $m \ge 1$ and suppose that $z^{m-1}$ satisfies (a)--(c).  To satisfy (a), define $z_n^m := z_n^{m-1}$ for all $n < m$. It remains to define $z_n^m$ for $n \ge m$. Let $w_0 := z_m^{m-1}$. For $k\ge 1$, given the pair $(w_{k-1}, z_{m+k}^{m-1})$, we inductively construct $(w_k, z_{m+k}^m)$ as follows, with the goal of defining $z_m^m$ as the limit of $\{w_k\}$:
\begin{description}[font=\normalfont]
     \item[Case I:] If there exists $x \in w_{k-1} \join z_{m+k}^{m-1}$ such that $u(x) \ge u(z_{m+k}^{m-1})$, then set
    \[
    w_k := w_{k-1}, 
    \qquad 
    z_{m+k}^m := x.
    \]
    \item[Case II:] If $u(x) < u(z_{m+k}^{m-1})$ for all $x \in w_{k-1} \join z_{m+k}^{m-1}$, then the upper pseudo quasi-supermodularity of $u$ implies that
    \[
    u(x) \ge u(w_{k-1}) 
    \quad \text{for all } x \in w_{k-1} \meet z_{m+k}^{m-1}.
    \]
    By \Cref{claim:nonempty-join}, there exists some
    \[
    x' \in (w_{k-1} \meet z_{m+k}^{m-1}) 
    \cap \{x : x \ge z_{m-1}^{m-1} = z_{m-1}^m\}.
    \]
    Define
    \[
    w_k := x',
    \qquad 
    z_{m+k}^m := z_{m+k}^{m-1}.
    \]
\end{description}
    In either case, for each $k \ge 1$, we have
\begin{gather}
\label{zmk}
z_{m-1}^m = z_{m-1}^{m-1} \le w_k \le w_{k-1}
\quad \text{and} \quad
w_k \le z_{m+k}^m,\\
\label{uwk}
u(w_k) \ge u(w_{k-1})
\quad \text{and} \quad
u(z_{m+k}^m) \ge u(z_{m+k}^{m-1}) \ge u(x_{m+k}).
\end{gather}

Since $\{w_k\}$ is a chain, it admits an infimum; define $z_m^m$ to be this infimum. By \eqref{zmk}, we have $z_m^m \le w_k \le z_{m+k}^m$ for all $k \ge 1$, so condition (b) holds for $z^m$. Moreover, by the order upper semicontinuity of $u$,
\[
u(z_m^m)
\ge \lim_{n \to \infty} \sup_{k \ge n} u(w_k)
\ge u(w_0)
= u(z_m^{m-1})
\ge u(x_m).
\]
Together with \eqref{uwk}, this shows that $z^m$ satisfies condition  (c), completing the induction.

We now  prove that $M_X (u)$ is a complete upper pseudo sublattice. Consider any $S \subset M_X (u)$ and any $\bar s \in \Join_X S$. We need to show $\bar s \in M_X (u)$, which will imply that $\bar s$ is a supremum of $S$ if $S$ is a chain (since $X$ is chain complete), so $M_X (u)$ is chain complete.
(That $\Meet_X S \subset M_X(u)$ follows from an analogous argument.)  Let $\hat M = M_X (u) \cap L_{\bar s}$.    Suppose for contradiction  that $\bar s \not\in  \hat M$.   
  By \Cref{claim:chain}, there is a collection $(C_x)_{x \in \hat M}$ such that  each $C_x$ is a maximal chain in $\hat M$ that contains $x$. Letting $z_x =\sup C_x$,  we have $z_x \in \hat  M$ by the order upper semicontinuity of $u$, which implies  $z_x < \bar s$.      Also, there must exist some  $x,y \in \hat M$ with  $z_x \ne z_y$ since otherwise we would have some $\bar x =\sup C_x, \forall x\in \hat M$  with $\bar x < \bar s$, which would imply $\bar x \in  U_{S}$  (since $S \subset \hat M$) and  contradict $\bar s \in \Join_X S$.    
Also,   we cannot have $z_x < z_y$ or $z_x > z_y$ since we could then  add $z_y$ to $C_x$ or $z_x$ to $C_y$ to form a larger chain, contradicting   the maximality of  $C_x$ or $C_y$, respectively.  
Thus, $z_x$ and $z_y$  must be  incomparable.  
By \Cref{claim:nonempty-join}, one can then find  some $\hat x \in  ( z_x \join z_y ) \cap L_{\bar s}$ with $\hat x > z_x$.  
Since $u (z_x) \ge u (x)$ for any $x \in  z_x \meet z_y$, the upper pseudo quasi supermodularity of $u$ implies $u (\hat x) \ge u (z_y)$ and thus $\hat x \in \hat  M$. This contradicts the maximality of $C_x$ since $C_x \cup \{ \hat x\}$ is a larger chain in $\hat M$. Thus, $\bar s \in \hat M$.  \end{proof}

\begin{proof}[Proof of \Cref{cor:random_choice}]
To begin, let us define
$\hat u: A \times \Delta (\Theta)\to \mathbb{R}$
by \begin{align}
    \label{hat-u} \hat u (a, \eta):=\int  u (a,\theta)  \eta (d \theta). 
\end{align}

First, $\hat u (\cdot, \eta)$  is  pseudo supermodular, and thus pseudo quasi-supermodular,  in $a$ since the pseudo supermodularity of $u$ is preserved under a convex combination.  By \Cref{lem:order_usc} in the Online Appendix,  $\hat u (\cdot,\eta)$ is order upper semicontinuous.  Thus, by \Cref{thm:pseudo-argmax}(ii), $\arg\max_{a \in A} \hat u (a, \eta) $ is a complete pseudo sublattice, admitting the largest and smallest points, $\overline a (\eta)$ and $\underline a (\eta)$.  Hence, any $a$ is not optimal if $a \not\le \overline{a}(\eta)$ or $a \not\ge \underline{a}(\eta)$. Letting $\overline{a} = \overline{a}(\eta)$ and $\underline{a} =\underline{a}(\eta)$, we show that $\delta_{\overline a (\eta)}$ and $\delta_{\underline a (\eta)}$ are the greatest and smallest elements in $\arg\max_{\tilde x \in X} \bar u (\tilde x, \eta)$, respectively.  Fix any  upward closed set $A' \subset A$ and any $x \in \arg\max_{\tilde x \in X} \bar  u (\tilde x, \eta)$.  Then,  
$\delta_{\underline a}(A') = x(A') = \delta_{\overline a}(A') = 1$ if $\underline a \in A'$; 
$\delta_{\underline a}(A') = 0 \le x(A') \le 1 = \delta_{\overline a}(A')$ if $\underline a \notin A'$ but $\overline a \in A'$; and 
$\delta_{\underline a}(A') = x(A') = \delta_{\overline a}(A') = 0$ if $\overline a\notin A'$. 
Thus, the desired conclusion follows.

Observe next that  for any $a'\ge a$ and $\eta'\ge^{sd} \eta$, 
\begin{align*}
 \hat u (a',\eta ')-\hat u(a,\eta') &=\int (u (a',\theta)- u(a, \theta))\eta'(d \theta)  \\
&\ge \int (u (a', \theta)- u (a, \theta)) \eta (d\theta)\\
&=   \hat u(a',\eta)-\hat u(a,\eta),
\end{align*} where the  inequality follows since   $\eta' \ge^{sd} \eta$ and
$u (a', \cdot)- u  (a, \cdot)$ is monotonic (due to the increasing differences property of $u$).  Thus, $\hat u$ satisfies single-crossing in $(a, \eta)$. By \Cref{cor:MCS-sublat}, $\arg\max_{a \in A} \hat u (a, \eta') \ge_{pss} \arg\max_{a \in A} \hat u (a,\eta)$, implying  that $\overline{a} (\eta)$ and $\underline{a} (\eta)$ are nondecreasing in $\eta$. Combined with the  observation that $\delta_{\overline a(\eta)}$ and $\delta_{\underline a(\eta)}$ are the largest and smallest elements of $M (\eta)$, this implies that $M (\eta') \ge_{ws} M(\eta)$. \end{proof}

\section{Proofs for \Cref{sec:FP}}

\begin{proof} [Proof of \Cref{thm:tarski-correspondence}]
Letting $X_F$ denote the set of fixed points of $F$, we first prove  $X_F  \ne \emptyset$.  To this end,  consider any point $x \in X_+ = \{x' \in X :  x'' \ge x'  \mbox{ for some } x'' \in F (x') \}.$\footnote{Note that $X_+$ is nonempty since it contains the smallest point of  $X$, which exists since $X$ is a complete pseudo lattice.}  By \Cref{claim:chain}, there is a maximal chain $C$ in $X_+$ that contains $x$. Letting $\bar x_C = \sup C $, we have $\bar y_C := \sup F (\bar x_C ) \ge \sup F (x') \ge x' $ for all $x' \in C$, which implies $\bar y_C \in U_C$ and thus $\bar y_C \ge \bar x_C$.  Then,  $F (\bar y_C )\ge_{pss} F (\bar x_C)$ and thus $y' \ge \bar y_C$ for some $y' \in F (\bar y_C)$ (since $\bar y_C \in F (\bar x_C)$). This implies $\bar y_C \in X_+$. If $\bar y_C > \bar x_C$, then $C \cup \{\bar y_C \}$ would be a larger chain in $X_+$ than $C$, contradicting the maximality of $C$. Thus, $\bar y_C =\bar x_C$, meaning $\bar x_C \in X_F$. 

To show that  $X_F$ is a complete pseudo lattice, we need to prove: (i) for  any $S \subset  X_F$, $\Join_{X_F} S \ne \emptyset$ and $\Meet_{X_F} S \ne \emptyset$; (ii)  $X_F$ is chain complete. 

For (i), we only prove  $\Join_{X_F} S \ne \emptyset$ (since proving $\Meet_{X_F} S \ne \emptyset$ is analogous).  Let $$ T= U_{S} \cap \{ x' \in X : x'\ge x'' \mbox{ for some } x''\in F (x') \}$$
and consider a maximal chain $C $ in $T$  (which is nonempty since it contains $\sup X$).  Letting $z: = \inf C$, we aim to show $z \in X_F$, which will imply  $z \in \Join_{X_F} S$ since, if there were any $ z' < z$ such that $z' \in U_S \cap X_F$, then  $C \cup \{ z' \}$ would form a larger chain in $T$ than $C$, a contradiction. Suppose now, for contradiction, that $z \not\in X_F$.  Observe that  $y:= \inf F (z) \le \inf F (x) \le x$ for all $x \in C$, implying $y \in L_C$ and thus $y \le z$.  We must have  $y < z$ since $z \not\in X_F$ and $y \in F (z)$.  Given this, we show below  that there is some $\tilde y  \in T$ with $\tilde y < z$, which will lead to the desired contradiction since $C \cup \{ \tilde y \}$ would be  a larger chain in $T$ than $C$.

By the well-ordering theorem, there exists an ordinal $\gamma$ such that $S = \{ s_\alpha \}_{\alpha < \gamma}.$ We construct a chain $\tilde C = \{ x_\alpha \}_{\alpha < \gamma}$ inductively, whose supremum gives the desired $\tilde y$.

\smallskip
\noindent
\emph{Initial step.}
Choose $x_1$ to be any element of $(y \join s_1) \cap L_z$, which is nonempty by \Cref{claim:nonempty-join} and the fact that $y, s_1 \in L_z$. Since $z \ge s_1$, $y \in F(z)$, and $s_1 \in F(s_1)$, the pSS monotonicity of $F$ implies $x_1 \in F(z),$
because $F(z) \ge_{pss} F(s_1)$.

\smallskip
\noindent
\emph{Inductive step.}
Let $\beta >1$ be any ordinal smaller than $\gamma$, and suppose that we have constructed an increasing chain
\[
(x_\alpha)_{\alpha < \beta} \subset L_z
\quad \text{with} \quad
x_\alpha \in F(z) \;\; \text{for all } \alpha < \beta.
\]
Define
\[
\tilde x := \sup_{\alpha < \beta} x_\alpha .
\]
This supremum exists and satisfies $\tilde x \in F(z)$ since $F(z)$ is a complete pseudo sublattice. As before, the set $(\tilde x \join s_\beta) \cap L_z$ is nonempty, and we choose any element of this set to be $x_\beta$.

Since $z \ge s_\beta$, $\tilde x \in F(z)$, and $s_\beta \in F(s_\beta)$, the pSS monotonicity of $F$ implies 
$x_\beta \in F(z),$
because $F(z) \ge_{pss} F(s_\beta)$. Adding $x_\beta$ to $(x_\alpha)_{\alpha < \beta}$ preserves the chain property.

\medskip
By construction, $\tilde C \subset F(z)$ and $x_\alpha \ge s_\alpha$ for all $\alpha < \gamma$. Since $\tilde C$ is a chain and $X$ is chain complete, there exists
\[
\tilde y := \sup \tilde C = \Join_X \tilde C .
\]
Moreover, $\tilde y \in F(z)$ because $F(z)$ is a complete pseudo sublattice. We also have $\tilde y \in U_S$, since $\tilde y \ge x_\alpha \ge s_\alpha$ for all $\alpha < \gamma$.

Finally, since $z \in U_{\tilde C}$ and $\tilde y = \sup \tilde C$, it follows that $\tilde y \le z$. In fact, $\tilde y < z$, because $z \notin F(z)$. By pSS monotonicity, this implies that there exists $\tilde z \in F(\tilde y)$ such that $\tilde z \le \tilde y$, since $\tilde y \in F(z)$ and $F(\tilde y) \le_{pss} F(z)$. Hence, $\tilde y \in T$, as desired.

For (ii), we prove only that any chain $C \subset X_F$ has a supremum in $X_F$, as the argument for the infimum is analogous. Let
\[
x_C := \sup_X C,
\]
which exists since $X$ is chain complete. Define a self-correspondence $G$ on $Y := U_{x_C}$ by
\[
G(x) := F(x) \cap Y \quad \text{for each } x \in Y .
\]

We first show that $G$ is nonempty-valued. It suffices to show that
\[
G(x_C) = F(x_C) \cap U_{x_C} \neq \emptyset .
\]
Indeed, once this holds, the pSS monotonicity of $F$ implies that for any $x \ge x_C$, there exist $x' \in F(x)$ and $x'' \in F(x_C)$ with $x'' \ge x_C$ such that, for some $\tilde x \in x' \join x''$, we have $\tilde x \in F(x)$. Since $\tilde x \ge x_C$, this implies $\tilde x \in F(x) \cap U_{x_C}$, and hence $G(x) \neq \emptyset$.

To show that $G(x_C) \neq \emptyset$, note that for each $x \in C$, we have $x \in F(x)$ and $F(x_C) \ge_{pss} F(x)$. It follows that $\sup F(x_C) \ge x$ for all $x \in C$. Therefore,
\[
\sup F(x_C) \in F(x_C) \cap U_C \subset F(x_C) \cap U_{x_C} = G(x_C).
\]

Observe next that for each $S \subset Y$, $\Meet_{Y} S  = (\Meet_X S) \cap Y$ and $\Join_Y S = (\Join_X S) \cap Y$. Using this observation, it is straightforward to see that $G$ is  complete-pseudo-sublattice-valued  and pSS monotonic.   Thus, letting $Y_G$ denote the set of fixed points of $G$  \jw{and applying the part (i) to $G$ and $Y_G$, it follows that $\Meet_{Y_G} Y_G \ne \emptyset$, so}  $Y_G$ contains  a smallest point, say $\underbar{$x$}_G$. Clearly, $\underbar{$x$}_G$ is also a fixed point of $F$ and corresponds to  $\sup_{X_F} C$ since any fixed point of $F$ weakly greater than $x_C$ is also a fixed point of $G$ and thus weakly greater than $\underbar{$x$}_G$.    \end{proof}

\begin{proof}[Proof of \Cref{thm:mcs of fp}]
Fix any $x \in \Fp (t)$ and let $Z = U_x$.  Define a self-correspondence on $Z$ as follows: for each $\tilde x\in Z$,  $H (\tilde x) = F (\tilde x, t') \cap Z $. Fix any $\tilde x\in Z$.  Since $x \in F (x, t)$ and $F(\tilde x, t') \ge_{pss} F(x,t') \ge_{uws} F (x, t) $, there is $y \in F (\tilde x, t')$ with $y \ge x$, so $H$ is nonempty-valued. It is straightforward to see that as $F (\tilde x, t')$ is a complete pseudo sublattice, so is $H (\tilde x) = F(\tilde x, t') \cap Z$.  It is also straightforward that $H$ is pSS monotonic on $Z$. Thus, by \Cref{thm:tarski-correspondence}, the set of fixed points of $H$, denoted $X_H$, is nonempty. Since $X_H \subset \Fp (t')$ and any $\tilde x\in  X_H$ is weakly greater than $x$, we have just proved that $\Fp (t') \ge_{uws}  \Fp (t)$. The proof for $\Fp (t') \ge_{lws}  \Fp (t)$ is analogous and hence omitted. \end{proof}

\begin{proof}[Proof of \Cref{prop:NE_existence_MCS}]
Denote   the (pure-strategy) best-response correspondence for player $i$ in game $\G (t)$ by  
\begin{align}
 	B_i(s_{-i},t):= \arg\max_{s_i\in S_i} u_i(s_i, s_{-i},t),\label{br-correspondence}
 	\end{align}
and let  $B(s,t)=\prod_{i \in I} B_i(s_{-i},t)$. Note that a strategy profile $s=(s_i)_{i \in I}$ is a (pure-strategy)  Nash equilibrium if and only if $s \in B(s,t)$. 

     To prove part (i), suppose that $\G(t)$ is pseudo quasi-supermodular.
Note first that $S_i$ is a complete pseudo lattice, and $u_i$ is order upper semi-continuous and pseudo quasi-supermodular in $s_i$ for each $i$ by assumption. So, by \Cref{thm:pseudo-argmax} (ii),  $B_i(s_{-i},t)$ is a nonempty, complete pseudo sublattice. Moreover, for each $s_{-i}, s_{-i}' \in S_{-i}$ with $s_{-i}' \ge s_{-i}$, because $u_i$ satisfies single-crossing in $(s_i,s_{-i})$  by assumption, by \Cref{cor:MCS-sublat}, $B_i(s_{-i}',t) \ge_{pss} B_i(s_{-i},t)$. Thus, it follows that $B(\cdot,t)$ has the property that $B(s',t) \ge_{pss} B(s,t)$ for every $s, s' \in S$ with $s' \ge s$, that is, $B(\cdot,t)$ is pSS monotonic (with respect to the product order on $S$). Therefore, it follows that $B$ 
is pseudo monotonic. Thus, by \Cref{thm:tarski-correspondence}, the set of fixed points of $B$ is a nonempty, complete pseudo lattice, and thus has a largest element. We complete the proof by observing that the set of fixed points of $B$ is equivalent to the set of pure Nash equilibria of $\G(t)$, $\Eq(t)$.

To prove part (ii), observe that since  $u_i$ satisfies single-crossing in $(s,t)$ for each $i \in I$ by assumption, by \Cref{cor:MCS-sublat}, we have $B (s,t') \ge_{pss} B(s,t)$ for each $s \in S$, which implies that $B (\cdot, t')$  weak-set  dominates $B (\cdot, t)$. So, by \Cref{thm:mcs of fp}, it follows that $\Eq(t') \ge_{ws} \Eq(t)$, as desired. 
\end{proof}

\begin{proof}[Proof of \Cref{thm:mcs-mixed}]
This result follows immediately from setting $\mu \equiv 0$ in \Cref{lem:perfect_eq_exist_comp} in \Cref{sec:wmcs of thpe}. 
\end{proof}

 \section{Proofs for \Cref{sec:wmcs of thpe}}
 \label{app:thpe}

First, we establish useful properties of the best response in the constrained games:  
\begin{lemma} \label{lem:MCS-constraind-mixed-alpha}  Consider a family of   pseudo supermodular games  $(I, S, u (\cdot, t))$. For any nonnegative measures $\mu = (\mu_i)_{i \in I} \in \mathcal{M}$, define the best response correspondence as
$$B_i^{\mu}(\sigma_{-i}; t)
:=\arg\max_{\sigma_i\in \Sigma_i^{\mu}} \int u_i(s_i, s_{-i},t) \sigma_i(ds_i)\sigma_{-i}(ds_{-i}).$$  Then,
\begin{itemize}
    \item[(i)] $B_i^{\mu}(\sigma_{-i}; t)$ has largest and smallest elements, $\sh_i^{\mu}(\sigma_{-i}; t)$ and  $\sl_i^{\mu}(\sigma_{-i}; t)$, which are both constrained-pure  and nondecreasing in $\sigma_{-i}$; 
    \item[(ii)] $\sh_i^\mu (\sigma_{-i}; t') \ge^{sd} \sh_i^\mu (\sigma_{-i}; t)$ and  $\sl_i^\mu (\sigma_{-i}; t') \ge^{sd} \sl_i^\mu (\sigma_{-i}; t)$ for any $t' \ge t$ if $u_i$ satisfies increasing differences in $(s_i,t)$. 
\end{itemize}  
\end{lemma}

\begin{proof}   Observe first that the unconstrained best response $B_i$, defined in \eqref{br_mixed_nash}, has the largest and smallest elements in pure strategies---denoted $\bh_i (\sigma_{-i};t)$ and $\bl_i (\sigma_{-i}; t)$, respectively---that are nondecreasing in $\sigma_{-i}$. This follows directly from \Cref{cor:random_choice} by  mapping the payoff function and best response in  \eqref{expected-utility-mixed} and \eqref{br_mixed_nash} to those in \eqref{expected-utility} and \eqref{expected-choice}, respectively, with  $a=s_i$ and $\theta = s_{-i}$ (so that   $x=\sigma_i$ and  $\eta=\sigma_{-i}$).   

Next, let   $\sh_i^{\mu}$ denote a strategy constrained-pure at $\bh_i:=\bh_i (\sigma_{-i};t)$: that is, for each $S_i' \subset S_i$,  
$$ \sh_i^{\mu} (S_i')= \begin{cases}
  1-\mu_i (S_i\setminus {S_i'})  &  \mbox{ if } \bh_i\in S_i' \\
   \mu_i (S_i')   & \mbox{ if } \bh_i\not\in S_i'. 
\end{cases} $$ 
Let us define for each $\sigma = (\sigma_i)_{i \in I}$ and $s_i \in S_i$, 
\begin{align*}
    \bar u_i (\sigma, t)  = \int_{s \in S} u_i (s, t) \sigma (ds) \; \mbox{ and } \;  \hat u_i (s_i, \sigma_{-i}, t)  = \int_{s_{-i} \in S_{-i}}   u_i (s_i, s_{-i},t) \sigma_{-i} (d s_{-i}).  
\end{align*} Then, for any $\sigma_i\in \Sigma_i^{\mu}$,
\begin{align*}
   & \bar u_i(\sh_i^{\mu}, \sigma_{-i},t)-\bar u_i(\sigma_i, \sigma_{-i},t)  \\ 
   = & \int \hat u_i(s_i, \sigma_{-i},t)\sh_i^{\mu}(d s_i)-\int \hat u_i(s_i, \sigma_{-i},t)\sigma_i(ds_i)\\
   = &  \hat u_i(\bh_i, \sigma_{-i},t)[\sh_i^{\mu}(\{\bh_i\})-\sigma_i(\{\bh_i\})]  - \int_{S_i\setminus \{\bh_i\}} \hat u_i(s_i, \sigma_{-i},t)[\sigma_i-\sh_i^{\mu}](ds_i)\\
   = & \int_{S_i\setminus \{\bh_i\}}[\hat u_i(\bh_i, \sigma_{-i},t)- \hat u_i(s_i, \sigma_{-i},t)][\sigma_i-\sh_i^{\mu}](ds_i)\\ 
    = & \int_{S_i\setminus \{\bh_i\}}[\hat u_i(\bh_i, \sigma_{-i},t)- \hat u_i(s_i, \sigma_{-i},t)][\sigma_i-\mu_i ](ds_i)\\ 
   \ge & 0.
\end{align*}
The third equality follows since
$\sh_i^{\mu}(\{\bh_i\})-\sigma_i(\{\bh_i\})=\int_{S_i\setminus \{\bh_i\}}[\sigma_i-\sh_i^{\mu}](ds_i)$ and the fourth equality follows from the construction of $\sh^{\mu}$.  Finally, the inequality follows from the fact that $\bh_i$ is a best reponse to $\sigma_{-i}$ and $\sigma_i\sqsupseteq \mu_i$ for any $\sigma_i\in \Sigma_i^{\mu}.$  Moreover, the inequality is strict if $[\sigma_i-\mu_i](S_i\setminus\{s_i: s_i\le \bh_i\})>0$.  This proves that $\sh_i^{\mu}\in B_i^{\mu}(\sigma_{-i}; t)$ and that any $\sigma_i\in \Sigma_i^{\mu}$ which  is not stochastically dominated by $\sh_i^{\mu}$ cannot be a best response.  We thus conclude that $\sh_i^{\mu}$ is the largest element of $B_i^{\mu}(\sigma_{-i}; t)$. Likewise,  a strategy constrained-pure at  $\bl_i (\sigma_{-i}; t)$, denoted as $\sl_i^\mu$, is the smallest point of $B_i^{\mu}(\sigma_{-i}; t)$. This completes the proof of part (i). 

To prove part (ii), observe first that as $u_i$ satisfies increasing differences in $(s_i,t)$, so does  $\hat u_i$ (since $\hat u_i$ is  a convex combination of $u_i$ with respect to $\sigma_{-i}$). Thus,    $\bh_i (\sigma_{-i}; t') \ge \bh_i (\sigma_{-i}; t)$ and $\bl_i (\sigma_{-i}; t') \ge \bl_i (\sigma_{-i}; t) $.   To prove  $\sh_i^{\mu}(\sigma_{-i};t')\ge^{sd} \sh_i^{\mu}(\sigma_{-i}; t)$, we need to show that $\sh_i^{\mu}(\sigma_{-i}; t') (S_i') \ge \sh_i^{\mu}(\sigma_{-i}; t) (S_i')$  for  any upward closed set $S_i' \subset S_i$.
There are two cases. Suppose first $\bh_i(\sigma_{-i}; t) \not\in S_i'$. Then, $\sh_i^{\mu}(\sigma_{-i}; t') (S_i') \ge  \mu_i (S_i' )= \sh_i^{\mu}(\sigma_{-i}; t) (S_i').$  Suppose next $\bh_i(\sigma_{-i}; t) \in S_i'$.  Then, $\bh_i(\sigma_{-i}; t') \in S_i'$  (since $S_i'$ is an upward closed set) so $\sh_i^{\mu}(\sigma_{-i}; t') (S_i') =\sh_i^{\mu}(\sigma_{-i}; t) (S_i') =  1- \mu_i (S_i\setminus S_i'),$ as desired.
Likewise, we have 
$\sl_i^{\mu}(\sigma_{-i}; t')\ge^{sd} \sl_i^{\mu}(\sigma_{-i}; t)$, completing our proof.  \end{proof}

\begin{proof}[Proof of \Cref{lem:perfect_eq_exist_comp}]  Throughout this proof, for any mapping $F$, let $X_F$ denote  its fixed-point set.

For part (i), consider the following self-maps on $\Sigma^\mu = \times_{i \in I} \Sigma_i^\mu$:
\begin{align}\label{fixed-map:F}
\overline F(\sigma) = (\sh_i^\mu (\sigma_{-i}; t))_{i \in I}, \quad
\underline F(\sigma) = (\sl_i^\mu (\sigma_{-i}; t))_{i \in I}, \quad
H(\sigma) = (B_i^\mu (\sigma_{-i}; t))_{i \in I},
\end{align} where $\sh_i^\mu$, $\sl_i^\mu$, and $B_i^\mu$ are as defined in \Cref{lem:MCS-constraind-mixed-alpha}.
 Note that any point in $X_{\underline F}$, $X_{\overline F}$, or $X_{H}$ is a Nash equilibrium of $\mG^\mu$. Both $\underline F$ and $\overline F$ are pseudo monotonic, as they are singleton-valued and $\sl_i^\mu (\cdot; t)$ and $\sh_i^\mu (\cdot; t)$ are weakly increasing by \Cref{lem:MCS-constraind-mixed-alpha}(i). By \Cref{thm:tarski-correspondence}, $X_{\underline F}$ and $X_{\overline F}$ each admit extremal points, which must be constrained-pure by construction. Let $\sl$ and $\sh$ denote the smallest and largest elements of $X_{\underline F}$ and $X_{\overline F}$, respectively. Moreover, by \Cref{lem:MCS-constraind-mixed-alpha}(i), $\sh_i^\mu (\sigma_{-i}; t) \ge^{sd} \sigma_i \ge^{sd} \sl_i^\mu (\sigma_{-i}; t)$ for any $\sigma_i \in B_i^\mu (\sigma_{-i};t)$, implying $\overline F (\sigma) \ge_{ws}   H(\sigma) \ge_{ws} \underline F (\sigma)$ at each $\sigma \in \Sigma^\mu$. Hence, by \Cref{thm:6prime}, $\sh \ge \sigma \ge \sl$ for every $\sigma \in X_{H} = \NE^\mu (t)$, as desired.

For part (ii), define $\overline F$ and $\underline{F}$ as in \eqref{fixed-map:F} and define \begin{align}\label{fixed-map:G}
\overline G(\sigma) = (\sh_i^\mu (\sigma_{-i}; t'))_{i \in I}, \quad
\underline G(\sigma) = (\sl_i^\mu (\sigma_{-i}; t'))_{i \in I}.
\end{align} By \Cref{lem:MCS-constraind-mixed-alpha}(ii),  $\overline{G}(\sigma) \ge_{ws} \overline{F}(\sigma)$ and $\underline{G}(\sigma) \ge_{ws} \underline{F}(\sigma)$ at each $\sigma \in \Sigma^\mu$. Then, by \Cref{thm:mcs of fp}, the largest  point of $X_{\overline{G}}$ is weakly greater than that of $X_{\overline{F}}$; the smallest point of $X_{\underline{G}}$ is weakly greater than that of $X_{\underline{F}}$. Given part (i), this implies $\NE^\mu(t') \ge_{ws} \NE^\mu (t)$.\end{proof}

\begin{proof}[Proof of \Cref{thm: main result for thpe}]
For parts (i) and (ii), we focus on the game $G (t)$.   Consider any sequence of constrained games $\{\mG^{\mu^n} (t)\}_n$ with  each $\mu^n$ belonging to $\M^0$ and $\| \mu^n \| \to 0$.    By  \Cref{lem:perfect_eq_exist_comp}(i), each constrained game $\mG^{\mu^n} (t)$ has a constrained-pure Nash equilibrium $\sigma^n$.  Then,  part (i)  follows from the next claim: \begin{claim} \label{claim:limit is pure}
A (sub)sequence of constrained-pure strategy profiles $\{\sigma^n\}_n$  must weakly converge to a pure strategy profile.  
\end{claim}
\begin{proof}
 Let $s_i^n$ denote  a pure strategy  on which $\sigma_i^n$ puts the maximum weight. By the compactness of $S_i$,   $\{s_i^n\}_n$ (or its subsequence) converges to some limit $\tilde s_i \in S_i$. To show that  $\sigma_i^n$ weakly converges to $\tilde s_i$  (i.e., the mixed strategy putting all the  weight on $\tilde s_i$), it suffices to show that for any bounded continuous function $f$ defined on $S_i$, \begin{align} \label{weak convergence}
        \int_{s_i \in S_i} (f (s_i) - f(\tilde s_i) ) d \sigma_i^n (s_i)  \to 0. 
    \end{align} To show this, for any $\epsilon > 0$, we can find $N$ such that for any $n > N$,  $| f (s_i^n) - f (\tilde s_i) | < \frac{\epsilon}{2}$ and $\mu_i^n (S_i) < \frac{\epsilon}{2 (M -m)}$, where $M =\sup_{s_i} f (s_i)$ and $m =\inf_{s_i} f (s_i)$.\footnote{We may assume without loss that $M > m$ since otherwise \eqref{weak convergence} holds trivially.} Note that the latter inequality implies $\sigma^n_i (S_i\setminus \{s_i^n\}) < \frac{\epsilon}{2 (M-m)}$. Observe now that 
    \begin{align*}
        \left|   \int_{s_i \in S_i} (f (s_i) - f(\tilde s_i) ) d \sigma_i^n (s_i) \right|
 & \le | f (s_i^n ) - f (\tilde s_i) | \sigma_i^n (\{s_i^n\})  +   \int_{s_i \in S_i \setminus \{ s_i^n\}} | f (s_i) - f(\tilde s_i) |  d \sigma_i^n (s_i)  \\  &  < \frac{\epsilon}{2} + (M -m)  \frac{\epsilon}{2(M-m)} =\epsilon,
    \end{align*} establishing the desired convergence. 
\end{proof}

 For part (ii), we first establish compactness. Since $\Sigma$ is a metrizable space, it suffices to prove sequential compactness.  
Consider a sequence of perfect equilibria $\{\sigma^n\}_n$. Because $\Sigma$ is compact, this sequence admits a convergent subsequence.  Denoting its limit by $\sigma$, we  show that $\sigma$ is a perfect equilibrium.

Since each $\sigma^n$ is a perfect equilibrium, there exists a sequence of constrained games $\{\mG^{\mu^{n,m}}(t)\}_m$ and an associated sequence of Nash equilibria $\{\sigma^{n,m}\}_m$ such that  $\| \mu^{n,m} \| \to 0$ and $\sigma^{n,m}\to \sigma^n$ as $m\to\infty$. For each $n$, choose $m(n)$ sufficiently large so that
\begin{align}
\label{pe-maximal-exist}
\max\bigl\{ \|\mu^{n,m(n)}\|,\; d(\sigma^n,\sigma^{n,m(n)}) \bigr\} < \frac{1}{n},
\end{align}
where $d(\cdot,\cdot)$ denotes the metric on $\Sigma$, which generates the weak topology.

For each $n$, define $\tilde\mu^n := \mu^{n,m(n)}$ and $\tilde\sigma^n := \sigma^{n,m(n)}$. Observe first that $\tilde\sigma^n$ is a Nash equilibrium of the constrained game $\mG^{\tilde\mu^n}(t)$. Moreover, by \eqref{pe-maximal-exist}, we have $\|\tilde\mu^n\|\to 0$ as $n\to\infty$, and $\tilde\sigma^n\to \sigma$. Indeed,
\[
d(\tilde\sigma^n,\sigma)= d(\sigma^{n,m(n)},\sigma)
\le d(\sigma^{n,m(n)},\sigma^n) + d(\sigma^n,\sigma),
\]
and both terms on the right-hand side converge to zero. This establishes that $\sigma$ is a perfect equilibrium.

To prove the existence of a maximal perfect equilibrium, note that by \Cref{lem:compact-chaincomp}, the set of perfect equilibria is chain complete since it is compact as shown above. Let $\sigma$ be any perfect equilibrium of $\G(t)$, whose existence follows from part (i). By \Cref{claim:chain}, there exists a maximal chain of perfect equilibria containing $\sigma$. Chain completeness then implies that this maximal chain admits a supremum, which must be a maximal perfect equilibrium.

 To prove that all maximal perfect equilibria are pure, it suffices to show that for any perfect equilibrium $\sigma$ of $\G(t)$, there exists a pure perfect equilibrium that weakly dominates $\sigma$ under $\ge^{sd}$. Consider a sequence of games $\{\mG^{\mu^n} (t)\}_n$ whose Nash equilibria $\{\sigma^n \}_n$ converge weakly to $\sigma$. By \Cref{lem:perfect_eq_exist_comp}(i), for each $\mG^{\mu^n}(t)$, there exists a constrained-pure strategy Nash equilibrium $\bar\sigma^n$ that dominates $\sigma^n$ under $\ge^{sd}$. By \Cref{claim:limit is pure}, $\{\bar\sigma^n \}_n$ (or its subsequence) converges to a pure strategy profile. We denote this limit by $\bar\sigma$, and it constitutes a perfect equilibrium.

To show that $\bar\sigma \ge^{sd} \sigma$, we use the following result from \cite{kamae1977stochastic}:
\begin{fact}
\label{fact:sd-dominance}
Let $\{P_n\}_n$ and $\{Q_n\}_n$ be sequences of probability measures on a Polish space that converge weakly to $P$ and $Q$, respectively. If $Q_n \ge^{sd} P_n$ for all $n$, then $Q \ge^{sd} P$.
\end{fact}
Since $\bar\sigma^n \ge^{sd} \sigma^n$ for all $n$, and $\bar\sigma$ and $\sigma$ are their respective weak limits, it follows that $\bar\sigma \ge^{sd} \sigma$, as desired. 

The proof for minimal perfect equilibria is analogous and is therefore omitted.

For part (iii), we only establish the upper weak-set dominance since the lower weak-set dominance can be established analogously.
Consider any   perfect equilibrium $\sigma$ for $\G (t) $ which is a weak  limit of  Nash equilibria $\{\sigma^n \}_n$ for $\{\mG^{\mu^n} (t) \}_n$. By \Cref{lem:perfect_eq_exist_comp}, there exists a constrained-pure strategy Nash equilibrium $\tilde \sigma^n$ for $\mG^{\mu^n} (t')$ that dominates $\sigma^n$ in $\ge^{sd}$.   Letting $\tilde \sigma$ denote a weak limit of $\{\tilde \sigma^n\}_n$, $\tilde \sigma $ is a perfect equilibrium of $\G (t')$  such that  $\tilde \sigma \ge^{sd} \sigma$ due to \Cref{fact:sd-dominance}. \end{proof}

\clearpage
\pagenumbering{arabic}
\renewcommand*{\thepage}{SA.\arabic{page}}
\renewcommand{\thelemma}{S\arabic{lemma}}
\setcounter{lemma}{0}
\begin{center}
    \huge{Supplementary Appendix for: \vskip 2pt Monotone Comparative Statics without Lattices}
\end{center}

\section{Omitted Results for \Cref{sec:prel} and 
\Cref{sec:ind-choice}}\label{sec:appendix-ind-choice}
\label{sec:appendix-ind-choice-omitted}

\begin{example}[Failure of Lattice Property under Convex Order]
\label{exa:failure_lattice}
Let $\Theta = \{\theta_1, \theta_2, \theta_3\}$, and let $\sigma = (\sigma_1, \sigma_2, \sigma_3) \in \Delta(\Theta)$ denote a belief.  
To simplify notation, we write $(\sigma_1, \sigma_2)$ for the belief  $(\sigma_1, \sigma_2, \sigma_3)$ with $\sigma_3 = 1 - \sigma_1 - \sigma_2 \ge 0$.  
Fix the prior belief at $\mu = (\tfrac{2}{5}, \tfrac{2}{5})$.  
The set 
\[
X_\mu = \{\, x \in X : \underline{x}_\mu \le^{cx} x \le^{cx} \overline{x}_\mu \,\}
\]
then includes, among others, the following measures:
\[
\begin{aligned}
P &= \tfrac{1}{2}\delta_{(\tfrac{3}{10},\,\tfrac{3}{10})}
     + \tfrac{1}{2}\delta_{(\tfrac{1}{2},\,\tfrac{1}{2})}, \\
Q &= \tfrac{1}{2}\delta_{(\tfrac{3}{10},\,\tfrac{1}{2})}
     + \tfrac{1}{2}\delta_{(\tfrac{1}{2},\,\tfrac{3}{10})}, \\
R &= \tfrac{1}{4}\big(\delta_{(\tfrac{1}{5},\,\tfrac{2}{5})}
     + \delta_{(\tfrac{2}{5},\,\tfrac{1}{5})}
     + \delta_{(\tfrac{3}{5},\,\tfrac{2}{5})}
     + \delta_{(\tfrac{2}{5},\,\tfrac{3}{5})}\big), \\
S &= \tfrac{5}{11}\,\delta_{(\tfrac{11}{20},\,\tfrac{11}{20})}
     + \tfrac{3}{11}\,\delta_{(\tfrac{11}{20},\,0)}
     + \tfrac{3}{11}\,\delta_{(0,\,\tfrac{11}{20})}.
\end{aligned}
\]
It is straightforward to verify that each of these measures is a mean-preserving spread of $\delta_\mu$.  
Following Proposition~4.5 of \citet{muller2006stochastic} (from which this example is adapted), one can check that:  
(i) both $R$ and $S$ are mean-preserving spreads of $P$ and $Q$;  
(ii) $R$ and $S$ are incomparable with each other; and  
(iii) there exists no $x \in X_\mu$ that is simultaneously a mean-preserving spread of $P$ and $Q$ and a mean-preserving contraction of $R$ and $S$.  
Hence, the least upper bound of $P$ and $Q$ does not exist under $\ge^{cx}$ (so $X_\mu$ is not a lattice).
\end{example}

\begin{proof}[Proof of \Cref{claim:chain}]
Fix any $X'\subset X$ and each $x\in X'$.  Consider the family $\pazocal C : = \{ C \subset X' : C \mbox{ is a chain with } x \in C   \}$,  and order the elements in $\pazocal C$ by the set inclusion order.   Take any chain of chains $\pazocal{D} = \{ C_\alpha   \}_{\alpha \in A} $  for totally ordered set $A$, that is, $ C_\alpha \in \pazocal C$, and  $ C_\alpha \subseteq  C_\beta$ or $ C_\beta \subseteq C_\alpha$  for all  $\alpha, \beta \in A $.  Set $C^{\sharp} = \midcup_{\alpha \in A} C_\alpha$, and observe that $C^\sharp \in \pazocal C$ and that it  is an upper bound  of $\pazocal D$ in $\pazocal C$.  Thus, by Zorn's lemma, $\pazocal C$ must contain a maximal element.
\end{proof}

\begin{proof}[Proof of \Cref{claim:nonempty-join}]  To prove the first statement, observe that by \Cref{claim:chain}, there is a maximal chain $ \hat C  $ in $S:=  U_z\cap L_{\{x,y\}}$ that contains $z$.  Letting  $\hat x =\sup \hat C$ by the chain-completeness, we have $\hat x \in U_z$.  To  argue that $\hat x \in x \meet y$, note first that since both $x$ and $y$ are upper bounds of $\hat C$,  $\hat x \le x$ and $\hat x \le y$, implying $\hat x \in S$. If $\hat x \not\in x \meet y$, then there exists $\bar x  > \hat x$ such that  $\bar x \in L_{\{x,y\}}$ so $\bar x \in S$.  Then,  $\hat C \cup \{\bar x\}$ would be  a larger chain, contradicting the maximality of $\hat C$.  The proof of the second statement is analogous and hence omitted. 
\end{proof}

\begin{lemma} \label{lem:order_usc}
For $\eta \in \Delta(\Theta)$, $\hat u(\cdot,\eta)$ is  order upper semicontinuous on $A$.
\end{lemma}

\begin{proof}
Let $C=\{a_\alpha\}_\alpha \subset A$ be any chain, and suppose that 
$a' \in (\Meet_A C)\cup(\Join_A C)$ exists (so $a'$ is either $\inf C$ or $\sup C$).

For each $\alpha$, define
\[
h_\alpha(\theta):=\sup_{\beta \ge \alpha} u(a_\beta,\theta).
\]
Then $(h_\alpha)_\alpha$ is pointwise decreasing in $\alpha$, and
\[
\inf_\alpha h_\alpha(\theta)
=
\inf_\alpha \sup_{\beta \ge \alpha} u(a_\beta,\theta)
\quad \text{for all } \theta.
\]

We obtain 
\begin{align}
\nonumber
\inf_\alpha \sup_{\beta \ge \alpha} \hat u(a_\beta,\eta)
&=
\inf_\alpha \sup_{\beta \ge \alpha}
\int u(a_\beta,\theta)\,\eta(d\theta) \\  \nonumber
&\le
\inf_\alpha 
\int \sup_{\beta \ge \alpha} u(a_\beta,\theta)\,\eta(d\theta) \\ \label{order_usc}
&=
\inf_\alpha 
\int h_\alpha(\theta)\,\eta(d\theta).
\end{align}

Since the family $(h_\alpha)_\alpha$ is bounded,
and  $(h_\alpha)_\alpha$ is pointwise decreasing, 
 the Monotone Convergence Theorem yields
\[
\inf_\alpha \int h_\alpha(\theta)\,\eta(d\theta)
=
\int \inf_\alpha h_\alpha(\theta)\,\eta(d\theta).
\]

Combining this with \eqref{order_usc},
\begin{align}
\label{order_usc-1}
\inf_\alpha \sup_{\beta \ge \alpha} \hat u(a_\beta,\eta)
\le
\int
\inf_\alpha \sup_{\beta \ge \alpha}
u(a_\beta,\theta)\,\eta(d\theta).\end{align}

By order upper semicontinuity of $u(\cdot,\theta)$ and the fact that 
$a'$ is $\inf C$ or $\sup C$,
\[
\inf_\alpha \sup_{\beta \ge \alpha}
u(a_\beta,\theta)
\le
u(a',\theta)
\quad \text{for all } \theta.
\]
Integrating both sides and using \eqref{order_usc-1} yield
\[
\inf_\alpha \sup_{\beta \ge \alpha} \hat u(a_\beta,\eta)
\le
\int u(a',\theta)\,\eta(d\theta)
=
\hat u(a',\eta).
\]

Thus $\hat u(\cdot,\eta)$ is order upper semicontinuous.
\end{proof}

\section{Omitted Proofs for \Cref{sec:FP}}

\subsection{Proofs for Directional Results}
\begin{proof}[Proof of \Cref{thm:5prime}]
For each $x \in X$, let $H(x)$ denote the largest element of $F(x)$, which exists since $F(x)$ is a complete upper pseudo sublattice. The upper pseudo monotonicity of $F$ implies that $H(x)$, as a singleton-valued correspondence,  is nondecreasing. Hence, by \Cref{cor:tarski}, the fixed-point set of $H$ is a nonempty complete pseudo lattice and therefore contains a largest element, denoted by $\bar x_H$.

To see that $\bar x_H$ is also the largest fixed point of $F$, let $F(\cdot,t)=F(\cdot)$ and $F(\cdot,t')=H(\cdot)$. Then $F(x,t') \ge_{ws} F(x,t)$ for all $x \in X$, which implies by \Cref{thm:6prime} that $\Fp(t') \ge_{ws} \Fp(t)$. It follows that $\bar x_H$ is the largest fixed point of $F$ as well.

The proof for the case in which the correspondence $F$ is lower pseudomonotonic is analogous and therefore omitted. 
\end{proof}

\begin{proof}[Proof of \Cref{thm:6prime}]
    We prove only the first statement, as the proof of the second statement is analogous.

Fix any $x \in \Fp(t)$ and let $Z := U_x$. Define two self-correspondences on $Z$. The correspondence $H$ is defined as in the proof of \Cref{thm:5prime}, while $G(\tilde x)=x$ for all $\tilde x \in Z$. Clearly, both $H$ and $G$ are pSS monotonic, and $H (\tilde x) \ge_{ws} G (\tilde x)$ for all $\tilde x \in Z$.

By \Cref{thm:mcs of fp}, the fixed-point set of $H$ weak-set dominates that of $G$, which implies that $H$ has a fixed point $x' \ge x$ (since $x$ is the unique fixed point of $G$). Observe that $x' \in \Fp(t')$, which establishes the desired result.  \end{proof}

\subsection{Omitted Proofs for Generalized Bertrand Games}
\label{app:generalized-bertrand}

First, we define directional versions of pseudo quasi-supermodular games. Specifically, 
we say that $\Gamma$ is an \textbf{upper (resp. lower) pseudo quasi-supermodular game} if, for all $i\in I$, condition (P1) of pseudo quasi-supermodular games holds and
\begin{enumerate}
    \item[(P2$''$)] $u_i$ is upper (resp. lower) pseudo quasi-supermodular in $s_i$ and satisfies upper (resp. lower) single-crossing in $(s_i,s_{-i})$.
\end{enumerate}
The following is a directional variant of \Cref{prop:NE_existence_MCS}.

\begin{manualprop}{2$'$}
 \label{prop:NE_existence_MCS-lower}  For a family of upper (resp. lower) pseudo quasi-supermodular   games  $\Gamma (t)$,
\begin{description} 
\item[(i)]  the set of (pure)  Nash equilibria $\Eq(t)$ is nonempty and has the largest (resp. smallest) element; 
\item[(ii)]  if $u_{i}$ satisfies upper (resp. lower) single-crossing in $(s_i,t)$ for all $i \in I$, then $\Eq(t') \ge_{uws} \Eq(t)$ (resp. $\Eq(t') \ge_{lws} \Eq(t)$) for all $t'>t$.
\end{description}
\end{manualprop}

\begin{proof}

First, we present a directional version of \Cref{cor:MCS-sublat}:

\begin{manualcorollary}{2$'$}
     \label{cor:MCS-sublat-upper}    If $u$ is upper (resp. lower) pseudo quasi-supermodular in $x$ and satisfies upper (resp. lower) single-crossing in $(x,t)$, then 
$M_{S'} (t') \ge_{upss} M_{S} (t)$ (resp. $M_{S'} (t') \ge_{lpss} M_{S} (t)$) for all $t' \ge t$ and $S'\ge_{pss} S$.
\end{manualcorollary}
\begin{proof}  For any  $x \in M_{S} (t)$ and $x' \in M_{S'} (t')$, choose  any $z \in x \meet x' $ and $z' \in x \join x'$. Since  $S'\ge_{pss} S$, $z\in S$ and $z'\in S'$.  Suppose that $x$ and $x'$ are incomparable (or else, the result would be trivial).  Letting $Z =\{ x, z \}$ and $Z' = \{x', z'\} $, we have $Z' \ge_{wpss} Z$, with $Z' \cap (x \join x') = \{z'\}$, and $Z \cap (x \meet x') =\{z\}$.  Since $M_{Z'} (t') \ge_{uwpss} M_{Z} (t) $ by \Cref{thm:MCS-UPSS} and since  $x \in  M_{Z} (t) $ and $x' \in M_{Z'} (t')$,  we must have   $z'\in M_{Z'} (t')$, which implies $z' \in M_{S'} (t')$. Since $z$ and $z'$ were chosen arbitrarily, we have  shown $M_{S'} (t') \ge_{upss} M_{S} (t)$.
\end{proof}

Next, we proceed to prove the ``upper'' versions of the statements of \Cref{prop:NE_existence_MCS-lower} (the proof for the ``lower'' version is symmetric and thus omitted).

     To prove part (i), suppose that $\G(t)$ is upper pseudo quasi-supermodular.
Note first that $S_i$ is a complete pseudo lattice, and $u_i$ is order upper semi-continuous and upper pseudo quasi-supermodular in $s_i$ for each $i$ by assumption. So, by \Cref{thm:pseudo-argmax-upper}  (ii),  $B_i(s_{-i},t)$ is a nonempty, complete upper pseudo sublattice. Moreover, for each $s_{-i}, s_{-i}' \in S_{-i}$ with $s_{-i}' \ge s_{-i}$, because $u_i$ satisfies upper single-crossing in $(s_i,s_{-i})$  by assumption, by \Cref{cor:MCS-sublat-upper}, $B_i(s_{-i}',t) \ge_{upss} B_i(s_{-i},t)$. Thus, it follows that $B(\cdot,t)$ has the property that $B(s',t) \ge_{upss} B(s,t)$ for every $s, s' \in S$ with $s' \ge s$, that is, $B(\cdot,t)$ is upper pSS monotonic (with respect to the product order on $S$). Therefore, it follows that $B$ 
is upper pseudo monotonic. Thus, by \Cref{thm:5prime}, the set of fixed points of $B$ is nonempty and admits a largest element. We complete the proof by observing that the set of fixed points of $B$ is equivalent to the set of pure Nash equilibria of $\G(t)$, $\Eq(t)$.

To prove part (ii), observe that since  $u_i$ satisfies upper single-crossing in $(s_i,t)$ for each $i \in I$ by assumption, by \Cref{cor:MCS-sublat-upper}, we have $B (s,t') \ge_{upss} B(s,t)$ for each $s \in S$, which implies that $B (s, t') \ge_{uws} B (s, t)$ for all $s$. So, by \Cref{thm:6prime}, it follows that $\Eq(t') \ge_{uws} \Eq(t)$, as desired. 
\end{proof}


Now we are ready to analyze generalized Bertrand games. We begin by establishing that a generalized Bertrand game is in the class of a directional version of pseudo quasi-supermodular games.

\begin{lemma} \label{br-lws} 
A generalized Bertrand game is lower pseudo quasi-supermodular.
\end{lemma}
\begin{proof}   
First, 
because the strategy space for each player is finite and totally ordered, it is a complete pseudo lattice. Second, each payoff function is order upper semicontinuous in the player's own strategy because the strategy space is finite. Moreover, it also trivially follows that 
the payoff function of each player is lower pseudo quasi-supermodular because $p_i \vee \bar p_i=\bar p_i$ and $  p_i \wedge  \bar p_i= p_i,$ for any $p_i$ and $\bar p_i$ with $ p_i<\bar p_i$.

Recall that the payoff function of firm $i$ is her profit given $p_i,p_{-i},$ as defined by \eqref{eq:bertrand_profit}, and fix any  $p_{-i}< p_{-i}'$ and pick any $p_i$ and $\bar p_i$ with $ p_i<\bar p_i$. We shall show that 
\begin{align}
U_i( p_i,p_{-i}')
\ge  
U_i(\bar p_i,p_{-i}')
\Rightarrow   
U_i( p_i,p_{-i})
\ge U_i(\bar p_i,p_{-i}),
\end{align}
which is the lower single-crossing condition we are left with to complete the proof.

Assume first that  $D_i (p_i, p_{-i}) =0$.
Then, (D1) implies  $D_i (\bar p_i, p_{-i}) =0$, so $U_i(p_i,p_{-i})=0= U_i(\bar p_i,p_{-i})$, as desired.

Assume next that  $D_i (p_i, p_{-i}) > 0$ and thus $D_i (p_i,p_{-i}') >0$ by (D1).  
Assume 
\begin{align}
 U_i(p_i, p_{-i}')\ge U_i(\bar p_i, p_{-i}'). \label{optimality-1}
\end{align} One can  define $c (p_{-i})$ and $K (p_{-i})$ such that  \begin{align}
   C_i (q)   = q  c  (p_{-i}) + K (p_{-i}) \mbox{ for }  q \in   \{  D_i (\bar p_i, p_{-i}), D_i (p_i,p_{-i}) \}. \label{cost-1}
\end{align}  Define similarly $c (p_{-i}')$ and $K (p_{-i}')$ by replacing $p_{-i}$ in \eqref{cost-1} with $p_{-i}'$.  By the convexity of $C_i$,     we have $c (p_{-i}') \ge c (p_{-i})$.\footnote{To see this, let $q = D_i (p_i, p_{-i})$ and $\bar q = D_i (\bar p_i, p_{-i})$ while letting $q' = D_i (p_i, p_{-i}')$ and $\bar q' = D_i (\bar p_i, p_{-i}')$. Note that $ q \ge \bar q$ and $q' \ge  \bar q'$. If $q = \bar q$, then $c (p_{-i})$ can be chosen sufficiency small to satisfy $c (p_{-i}) \le c (p_{-i}')$. Also if $q' =\bar q'$, then $c (p_{-i}')$ can be chosen sufficiency large  to satisfy $c (p_{-i}) \le c (p_{-i}')$.  Suppose thus that $q > \bar q$ and $q' > \bar q'$.   Observe now that, by (D1), $q \le  q'$ and $\bar q \le \bar q'$. Given this, the convexity of $C_i$ implies $$c(p_{-i}) =\frac{C_i (q) -C_i (\bar q) }{q -\bar q} \le \frac{C_i (q') - C_i (\bar q')}{q' - \bar q'} = c (p_{-i}').$$}
  Observe that \eqref{optimality-1} can be rewritten as   \begin{align} 
\label{optimality-2}    (p_i - c (p_{-i}')) D_i (p_i, p_{-i}') \ge (\bar p_i -c (p_{-i}')) D_i (\bar p_i, p_{-i}'). 
\end{align}   We next argue that $p_i -c (p_{-i}') \ge 0$. This is immediate from \eqref{optimality-2} if $D_i (\bar p_i, p_{-i}') =0$  (recall  $D_i (p_i, p_{-i}') >0$). Suppose thus that $D_i (\bar p_i, p_{-i}') > 0$. If  $p_i - c (p_{-i}') < 0$, then \eqref{optimality-2} would imply   $$ \bar p_i - c (p_{-i}') \le  (p_i - c (p_{-i}')) \frac{D_i (p_i, p_{-i}')}{D_i (\bar p_i, p_{-i}')} \le p_i -c (p_{-i}') $$ since $\frac{D_i (p_i, p_{-i}')}{D_i(\bar p_i, p_{-i}')} \ge 1 $, which contradicts $\bar p_i > p_i$. Thus $\bar p_i  - c (p_{-i}') > p_i - c(p_{-i}') \ge 0$.  Using this and $c(p_{-i}') \ge c (p_{-i})$, we obtain \begin{align}
 \frac{p_i - c (p_{-i})}{\bar p_i - c (p_{-i})} \ge \frac{p_i - c (p_{-i}')}{\bar p_i - c (p_{-i}')} \ge \frac{D_i (\bar p_i, p_{-i}')}{D_i (p_i,p_{-i}')} \ge \frac{D_i (\bar p_i, p_{-i})}{D_i (p_i,p_{-i})}, 
\label{optimality-3}\end{align} where the second inequality follows from \eqref{optimality-2} while the last inequality from (D2).  It follows from \eqref{optimality-3} that $ (p_i - c (p_{-i})) D_i (p_i, p_{-i}) \ge (\bar p_i -c (p_{-i})) D_i (\bar p_i, p_{-i})$, which implies $U_i (p_i, p_{-i})\ge U_i (\bar p_i, p_{-i}) $, completing the proof. \end{proof}

Next, we proceed to study the comparative statics of generalized Bertrand games. 

 \begin{lemma} \label{bertrand-br-mcs} Suppose that a  family of generalized Bertrand games $\Gamma (t)$ satisfies (B1) and (B2). Then, for any $i \in I$,  $U_i$ satisfies lower single-crossing in $(p_i,t)$.
\begin{proof} Recall that the shift from $\G (t)$ to $\G (t')$ with $t'>t$ involves the two changes, (B1) and (B2).  It suffices to establish the result under each change separately. First,  consider the case in which (B2) holds, while the demand $D_i(p,t)$ is constant in $t$. Given (B2), we shall show that $U_i$ satisfies lower single-crossing in $(p_i,t)$: that is, for any $p_i' \ge p_i$ and $p_{-i}$, $U_i (p_i', p_{-i},t) - U_i (p_i,p_{-i},t) > 0$ implies $U_i (p_i', p_{-i},t') - U_i (p_i,p_{-i},t') > 0$. To show it, note that
\begin{align}
\nonumber & U_i (p_i', p_{-i},t) - U_i (p_i,p_{-i},t) > 0  \\ 
\nonumber \Leftrightarrow  \;  & p_i'D_i(p_i',p_{-i},t)-C_i(D_i(p_i',p_{-i},t),t)-[p_iD_i(p_i,p_{-i},t)-C_i(D_i(p_i,p_{-i},t),t)]>0 \\
 \label{eq:bertrand_cs} \Leftrightarrow  \;  & p_i'D_i(p_i',p_{-i},t)-p_iD_i(p_i,p_{-i},t)> C_i(D_i(p_i',p_{-i},t),t)-C_i(D_i(p_i,p_{-i},t),t).
\end{align}
Because $ C_i(q',t)-C_i(q,t) \le C_i(q',t')-C_i(q,t')$ for every $q'>q$ by assumption (B2) and $D_i(p_i',p_{-i},t) \le D_i(p_i,p_{-i},t)$ as $p_i' \ge p_i$,  we have 
\begin{align*}
    C_i(D_i(p_i',p_{-i},t),t)-C_i(D_i(p_i,p_{-i},t),t) \ge C_i(D_i(p_i',p_{-i}),t')-C_i(D_i(p_i,p_{-i}),t').
\end{align*}
Therefore, it follows  from \eqref{eq:bertrand_cs} that that 
\begin{align*}
    p_i'D_i(p_i',p_{-i},t)-p_iD_i(p_i,p_{-i},t)>  C_i(D_i(p_i',p_{-i},t),t')- C_i(D_i(p_i,p_{-i},t),t'),
\end{align*} which is equivalent to 
\begin{align*}
     U_i (p_i', p_{-i},t') -  U_i (p_i,p_{-i},t') > 0,
\end{align*}
as desired. 

Next, consider the case in which (B1) holds while $C_i(q,t)$ is constant in $t$. Showing the lower single-crossing of $U_i$ in $(p_i,t)$ for this case is analogous to the proof of \Cref{br-lws} and hence omitted.  
\end{proof}
\end{lemma}

\begin{proof}[Proof of \Cref{cor:general_bertrand}]  The existence of Nash equilibrium is a direct consequence of  \Cref{prop:NE_existence_MCS-lower} (i) and \Cref{br-lws}. The comparative statics result between the sets of equilibria in $\G (t)$ and $\G (t')$ follows directly from \Cref{prop:NE_existence_MCS-lower} (ii) and  \Cref{bertrand-br-mcs}.

To prove the comparative statics result on the firm profit, suppose that $\Gamma (t)$ and $\Gamma (t')$ satisfy (B1) and (B2$'$) and consider  any Nash equilibrium $p' =( p_i')_{i\in I}$ in $ \G (t')$. By \Cref{br-lws}, \Cref{bertrand-br-mcs} and \Cref{prop:NE_existence_MCS-lower}(ii), there exists an equilibrium $p^* \le  p'$ in $\G (t)$. Now, 
consider any firm $i$ with $c_i'= c_i$. 

First, suppose that $ p_i' <c_i$. Then, it follows that $D_i(p', t')=0$ because otherwise $U_i(p', t')=(p_i ' -c_i) D_i(p', t' )<0 \le ( \max P_i-c_i)  D_i(\max P_i, p_{-i}',t')=U_i(\max P_i,p_{-i}', t')$ because $\max P_i \ge c_i$ by assumption, contradicting the assumption that $p'$ is a Nash equilibrium in $\G (t')$.\footnote{We denote $ U_i(p, t' )=(p_i-c_i)  D_i(p, t')$ for each $i \in I$ and $p \in P$.} This implies that $U_i(p',t')=0$. Now, we will show that $U_i(p^*, t)=0$. To show this, suppose for contradiction that $U_i(p^*,t) \neq 0.$ Because $p^*$ is an equilibrium and $\max P_i \ge c_i$, this implies that $D_i(p^*,t)>0$, $p^*_i > c_i$, and $U_i(p^*,t)>0$. Hence, it follows that
\begin{align*}
0 & < U_i(p^*, t)\\
& = (p^*_i-c_i) D_i(p^*,t) \\
& \le (p^*_i-c_i) D_i(p^*, t') \\
& \le (p^*_i-c_i) D_i(p^*_i, p_{-i}', t') \\
& \le  U_i(p', t' ),
\end{align*}
where the first inequality is as established earlier, the equality is by the definition of $U_i$, the second inequality follows from condition (B2$'$) and $p^*_i >c_i$, the third inequality follows from $p^* \le p'$, condition (D1), and $p^*_i >c_i$, and the last inequality follows from the assumption that $p'$ is a Nash equilibrium in $\G (t')$. Thus we obtain $U_i(p', t')>0$, a contradiction. Thus we have shown that $U_i(p^*,t)=0= U_i(p', t')$.

Second, suppose that $p_i' \ge c_i$. Then, first note that, for any $p_i \in P_i$  with $p_i \ge c_i$, $U_i (p_i,p_{-i},t) = (p_i -c_i) D_i (p_i, p_{-i},t)$ is weakly increasing in $p_{-i}$. Define  $\Pi_i (p_{-i},t) := \max_{p_i\in P_i}  U_i (p_i, p_{-i},t)$, and define $\Pi_i (p_{-i},t')$ similarly. Note that  the above monotonicity of $U_i$  for $p_i \ge c_i$ implies $\Pi_i (\cdot,t)$ is weakly increasing and that for any $p_{-i}$  and $p_i \ge c_i$,  $(p_i -c_i) D_i (p_i, p_{-i},t') \ge (p_i -c_i) D_i (p_i, p_{-i},t)$ and thus $\Pi_i (p_{-i},t') \ge \Pi_i (p_{-i},t)$.
 Therefore, it follows that for each $i$ with $c_i' = c_i$,  $\Pi_i (p_{-i}^*, t) \le \Pi_i (p_{-i}',t)  \le \Pi_i  (p_{-i}', t')$, where the first inequality follows because $p^* \le p'$ and $\Pi_i(\cdot, t )$ is weakly increasing as established earlier, and the second inequality follows because $\Pi_i (p_{-i}, t') \ge \Pi_i (p_{-i},t)$ for all $p_{-i}$ as established earlier as well. 

The preceding two cases complete the proof. 
\end{proof}

Finally, we demonstrate that the class of our generalized Bertrand games indeed subsumes pure Bertrand games as special cases.

\begin{lemma} \label{lem:pure-bertrand} A pure  Bertrand game is a generalized Bertrand game. 
\end{lemma}
\begin{proof} 
	It suffices to check  (D2) since it is straightforward to check   (D1).  Fix any  $p_i<p_i', p_{-i}<p_{-i}'$ such that $D_i(p_i,p_{-i})>0$.  It must be that 
$p_i\le p_{-i}^m:=\min_{j\ne i}p_j$.   There are two cases.

Consider first $p_i=p_{-i}^m$.  Then, $D_i(p_i', p_{-i})=0$ and $D_i(p_i, p_{-i}')>0$.  Hence, 
$$\frac{D_i(p_i', p_{-i})}{D_i(p_i, p_{-i})}=0\le \frac{D_i(p_i', p_{-i}')}{D_i(p_i, p_{-i}')}.$$

Consider next $p_i<p_{-i}^m$, so $D_i(p_i, p_{-i})=1$.  By (D1), $D_i(p_i, p_{-i}')=1$.
  Hence,
$$\frac{D_i(p_i', p_{-i})}{D_i(p_i, p_{-i})}\le \frac{D_i(p_i', p_{-i}')}{D_i(p_i, p_{-i}')} \Leftrightarrow D_i(p_i', p_{-i})\le D_i(p_i', p_{-i}').$$ The latter inequality is  a direct consequence of  (D1). 
\end{proof}

\section{Omitted Proof for \Cref{sec:wmcs of thpe}}

\begin{proof}[Proof of \Cref{lem:perfect is nash}]   We first establish  the following claim that shows the Hausdorff distance between $\Sigma_i$ and $\Sigma_i^{\mu_i^n}$ goes to zero as $n \to \infty$: 
\begin{claim} \label{claim:perturbed-distance}
    Fix any $\epsilon >0$. For sufficiently large $n$, 
   $  \sup_{\sigma_i ' \in \Sigma_i} \inf_{\tilde \sigma_i  \in \Sigma_i^{\mu_i^n}} d ( \tilde \sigma_i, \sigma_i' ) < \epsilon$.    
\end{claim}
\begin{proof}
    Given any $\sigma_i' \in \Sigma_i$,  let 
$\tilde \sigma_i  =  (1-\mu^n_i (S_i) )  \sigma_i'  +    \mu^n_i.  $ Clearly, $\tilde \sigma_i\in \Sigma_i^{\mu_i^n}$. Also, 
\begin{align*}
     d ( \tilde \sigma_i, \sigma_i'  )  =   \sup_{S_i'\in \mathcal S_i} | - \mu^n_i (S_i)    \sigma_i' (S_i') +    \mu^n_i (S_i') | \le \sup_{S'\in \mathcal S_i}    \mu^n_i (S_i)     \sigma_i' (S_i')  +    \mu^n_i (S_i') \le 2 \mu_i^n (S_i).
\end{align*}  
The proof is complete by choosing $n$ sufficiently large so that $\mu_i^n (S_i) < \frac{\epsilon}{2}$. 
\end{proof}
Consider a perfect equilibrium $\sigma$ of game $\mG=(I, \Sigma, \bar u)$.  Thus, there exists a sequence  $\sigma^n $ of Nash equilibrium of $\mG^{\mu^n}$ that weakly converges to $\sigma$.  Fix  any  $i\in I$ and $\sigma_i' \in \Sigma_i$ and  use  \Cref{claim:perturbed-distance} to  find a sequence   $\tilde \sigma_i^n \in \Sigma_i^{\mu_i^n}$ such that $ d(\tilde  \sigma_i^n , \sigma_i' ) \to 0$ so $\tilde \sigma^n_i$ weakly converges to $\sigma_i'$. Since $\sigma^n $ is a Nash equilibrium of $\mG^{\mu^n}$, we have $\bar u_i (\sigma_i^n, \sigma_{-i}^n ) \ge \bar u_i (\tilde \sigma_i^n, \sigma_{-i}^n)$.   Since $u_i$ is continuous, $\bar u_i (\sigma_i, \sigma_{-i}) \ge \bar u_i (\sigma_i', \sigma_{-i})$ by the weak convergence of $\sigma^n$ and $\tilde \sigma_i^n$ to $\sigma$ and $\sigma_i'$, respectively.\footnote{Note that since  $\Sigma_i$ is endowed  with the weak topology,  $\bar u_i (\cdot)$ remains continuous.} Thus, $\sigma$ is a Nash equilibrium of $\mG$ as desired. \end{proof}

\begin{lemma} \label{lem:constrained game st space}
If $S_i$ is a complete pseudo lattice,  so is $\Sigma_i^\mu$. 
\end{lemma}  
\begin{proof}
Thanks to \Cref{lem:compact-chaincomp} and \Cref{cor:compact-with-ext}, it suffices to show that $\Sigma_i^{\mu}$ is a compact set with the largest and smallest elements and that the stochastic dominance order on $\Sigma_i^\mu$ is closed under  the weak topology on $\Sigma_i^\mu$.  

To first show the compactness, we prove that $\Sigma_i^\mu$ is closed. Consider any sequence $(\sigma_i^n)_n $ in $\Sigma_i^\mu$ so that $\sigma_i^n (S') \ge \mu_i (S'),\forall n, \forall  S' \in \mathcal S_i$. Letting $\r'$ denote  the weak limit of the sequence (which must exist since $\Sigma_i$ is compact under the weak topology), we must show $\sigma_i' (S') \ge \mu_i (S'), \forall S' \in \mathcal S_i$. To show this, observe first that for any closed set $\tilde S \subset S_i$, 
\begin{align}
    \label{comp-closed} \sigma_i' (\tilde S) \ge  \lim\sup_{n} \sigma_i^n (\tilde S) \ge \mu_i (\tilde S)
\end{align}  (by the weak convergence of $(\sigma_i^n)_n$ to $\sigma_i'$   and Portmanteau theorem).   Observe next that 
 every Borel measure $\sigma_i'$ on Polish space  $S_i$  is regular so that   for any  $S' \in\mathcal S_i$, \begin{align}
    \label{measure-approx-compact} \sigma_i' (S') = \sup \{ \sigma_i' (\tilde S) : \tilde  S \subset S' \mbox{ and } \tilde S \mbox{ is compact}  \}.
\end{align} Since every compact set $\tilde S$ is closed, we have $\sigma_i' (\tilde S)  \ge \mu_i (\tilde S)$ by \eqref{comp-closed}, which implies  by \eqref{measure-approx-compact} that for any  $S' \subset \mathcal S_i$,  $\sigma_i' (S') \ge \mu_i (S') $.

The existence of smallest point $\underbar{$\sigma$}_i^{\mu}$ is proved by construction as follows: 
letting $\underbar{$s$}_i$ denote the smallest element of $S_i$ (which exists since $S_i$ is a complete pseudo  lattice),  
 $$\underbar{$\sigma$}_i^{\mu}(S_i'):=\begin{cases}
   1-\mu_i(S_i\setminus S_i')  &  \mbox{ if } \underbar{$s$}_i \in S_i' \\
   \mu_i (S_i')   & \mbox{ otherwise}.  
\end{cases} $$   To show that $\sigma_i \ge^{sd} \underbar{$\sigma$}_i^{\mu}$ for any $\sigma_i \in \Sigma_i^\mu$,  consider  any upward closed set $S_i' \in \mathcal{S}_i$. If $\underbar{$s$}_i \in S_i' $, then  $S_i' = S_i$ and thus $\underbar{$\sigma$}_i^{\mu}(S_i') =1  =\sigma_i (S_i')$.  Otherwise,  $\underbar{$\sigma$}_i^{\mu}(S_i')  = \mu_i (S_i') \le \sigma_i (S_i')$, as desired. The largest point obtains analoguously.  

Lastly, the closedness of the stochastic dominance order follows immediately from  \Cref{fact:sd-dominance}.
\end{proof}

\end{document}